\documentclass[11pt,letterpaper]{article}

\usepackage[top=1in,bottom=1in,left=1in,right=1in]{geometry}

\usepackage{amsmath,amssymb}

\usepackage{amsthm}
\usepackage{mathtools}
\usepackage{titlesec}
\usepackage{dsfont}
\usepackage{makecell}
\usepackage{changepage}
\usepackage{setspace}
\usepackage{multirow}
\usepackage{comment}
\usepackage{xr}
\usepackage{todonotes}
\usepackage{graphicx}
\usepackage{svg}
\usepackage{subcaption}
\usepackage{float}
\usepackage{pict2e}
\usepackage{color}
\usepackage[title]{appendix}
\usepackage{enumerate}

\newtheorem{lemma}{Lemma}
\newtheorem{proposition}{Proposition}[section]
\theoremstyle{definition}

\newtheorem*{proposition*}{Proposition}
 \newcommand{\beq}{\begin{eqnarray}}
 \newcommand{\beqq}{\begin{eqnarray*}}
 \newcommand{\eeq}{\end{eqnarray}}
 \newcommand{\eeqq}{\end{eqnarray*}}
 \newcommand{\pmtx}{\begin{pmatrix}}
 \newcommand{\pmtrx}{\end{pmatrix}}

\usepackage{textcomp,marvosym}

\usepackage{enumitem}
\newcommand{\itemgo}{\begin{itemize}}
\newcommand{\itemend}{\end{itemize}}
\usepackage{hyperref}
\newcommand{\real}{\mathrm{Real}}

\newcommand{\R}{\mathbb{R}}
\newcommand{\osc}{\mathrm{osc}}
\newcommand{\cR}{\mathcal{R}}
\newcommand{\cE}{\mathcal{E}}
\newcommand{\cQ}{\mathcal{Q}}
\newcommand{\sgn}{\mathrm{sgn}}

\begin{document}

\begin{center}
\textbf{\large Optimal control of ribosome population for gene expression under periodic nutrient intake}

\vspace{7mm}

Clément Soubrier\textsuperscript{1}, Eric Foxall\textsuperscript{2}, Luca Ciandrini\textsuperscript{3,4}, Khanh Dao Duc\textsuperscript{1}
\end{center}
\vspace{5mm}
{\small
$^{1}$ Department of Mathematics, University of British Columbia, Vancouver, BC V6T 1Z4, Canada\\
$^{2}$ Department of Mathematics, University of British Columbia Okanagan, Kelowna, BC V1V 1V7, Canada\\
$^{3}$ Centre de Biologie Structurale (CBS), Univ Montpellier, CNRS, INSERM, Montpellier 34090, France\\
$^{4}$ Institut Universitaire de France (IUF)
}
\vspace{5mm}

%%%%%%%%%%%%%%%%%%%%%%%%%%%%%%%%%%%%%%%%%%%%%%%%%%%%%%%%%%%%%%%%%%%%%%%%%%
\begin{abstract}
Translation of proteins is a fundamental part of gene expression that is mediated by ribosomes. As ribosomes significantly contribute to both cellular mass and energy consumption, achieving efficient management of the ribosome population is also crucial to metabolism and growth. Inspired by biological evidence for nutrient-dependent mechanisms that control both ribosome active degradation and genesis, we introduce a dynamical model of protein production, that includes the dynamics of resources and control over the ribosome population. Under the hypothesis that active degradation and biogenesis are optimal for maximizing and maintaining protein production, we aim to qualitatively reproduce empirical observations of the ribosome population dynamics. Upon formulating the associated optimization problem, we first analytically study the stability and global behaviour of solutions under constant resource input, and characterize the extent of oscillations and convergence rate to a global equilibrium. We further use these results to simplify and solve the problem under a quasi-static approximation. Using biophysical parameter values, we find that optimal control solutions lead to both control mechanisms and the ribosome population switching between periods of feeding and fasting, suggesting that the intense regulation of ribosome population observed in experiments allows to maximize and maintain protein production. Finally, we 
find some range for the control values over which such a regime can be observed, depending on the intensity of fasting. 
\end{abstract}
%%%%%%%%%%%%%%%%%%%%%%%%%%%%%%%%%%%%%%%%%%%%%%%%%%%%%%%%%%%%%%%%%%%%%%%%%%

%%%%%%%%%%%%%%%%%%%%%%%%%%%%%%%%%%%%%%%%%%%%%%%%%%%%%%%%%%%%%%%%%%%%%%%%%%
\section{Introduction}
%%%%%%%%%%%%%%%%%%%%%%%%%%%%%%%%%%%%%%%%%%%%%%%%%%%%%%%%%%%%%%%%%%%%%%%%%%

Expression of genes is one of the most essential biological processes underlying life. One fundamental part of this process is the translation of proteins from messenger RNAs, mediated by molecular machines called ribosomes. Under optimal growth conditions, ribosomes massively contribute to both cellular dry mass (30\%) and energetic consumption (60\%) in \textit{E.coli} \cite{li2014quantifying}, suggesting that efficient management of the ribosomal population is crucial for controlling metabolism and growth, according to the amount of nutrients available \cite{bosdriesz2015fast}. Such a coupling has been widely investigated for the past few years to derive quantitative relationships between growth rate and resource allocation in bacteria \cite{bosdriesz2015fast, scott2010interdependence, scott_emergence_2014, dai2020coupling, roy2021unifying, Calabrese2023}. At the molecular level, these laws depend on different pathways that can either increase the ribosome population at fast growth rate, or keep a low level of actively translating ribosome under slower growth conditions
\cite{dai_reduction_2017, wu2022cellular}. In eukaryotes, similar relationships and nutrient-dependent mechanisms 
have been found to regulate the production or degradation of ribosomes, which could have various implications for cell growth, disease and cancer \cite{
dai2020coupling, metzl2017principles, xia2022proteome}. 

Although often neglected in the literature, protein degradation also plays a role in protein allocation models~\cite{calabrese2022protein, gupta2022global}. 
Actively targeting ribosomes for degradation, such as through the ribophagy pathway \cite{wyant2018nufip1}, allows for potential recycling of rare amino acids from ribosomal components \cite{wyant2018nufip1, nakatogawa2018spoon}. These mechanisms might have a role in many physiological contexts, also in higher organisms. For instance, resource recycling might be connected with the fluctuations of ribosome number in mice livers in response to daily rhythms, influenced by feeding-fasting and light-dark cycles~\cite{sinturel2017diurnal}. However, the relationship between resource recycling and quantitative laws has received limited attention in the field. Both theoretical modeling and quantitative experiments have yet to extensively explore the connection between resource recycling and the allocation of ribosomal resources.

In this paper, our main goal is to study whether the oscillatory patterns observed in the ribosomal populations can be explained by considering the control mechanisms associated with gene expression and nutrient availability. To do so, we introduce here a dynamical system of protein production, that includes resources used for the production of both ribosomal and non-ribosomal proteins. This model also features feedback and control functions that modulate the production and degradation of ribosomes and allow for some recycling of resources. By assuming that ribosome population is mainly controlled to stabilize and maximize expression of certain genes, we formulate an optimization problem under periodic conditions of both low and high nutrient input. To solve it, we first study the global stability and behaviour of trajectories of the system, and quantify the asymptotic convergence speed towards a global equilibrium. This result allows us to find conditions under which we can use a quasi-static approximation of the system, to derive an analytical threshold for the existence of optimal controls. Our numerical application with biophysical parameters indicates that during feeding-fasting cycles, the model can reproduce oscillations in agreement with experimental observations while maintaining stable gene expression.   As a result, we establish that recycling of ribosomes and reduction of ribogenesis can allow for optimal constant protein production during periods of fasting and feeding, and we establish the exact range for the control parameters to guarantee such a regime.

%%%%%%%%%%%%%%%%%%%%%%%%%%%%%%%%%%%%%%%%%%%%%%%%%%%%%%%%%%%%%%%%%%%%%%%%%%
\section{Mathematical modeling of protein production and optimal control in eukaryotes} \label{section:model}
%%%%%%%%%%%%%%%%%%%%%%%%%%%%%%%%%%%%%%%%%%%%%%%%%%%%%%%%%%%%%%%%%%%%%%%%%%

%%%%%%%%%%%%%%%%%%%%%%%%%%%%%%%%%%%%%%%%%%%%%%%%%%%%%%%%%%%%%%%%%%%%%%%%%%
\subsection{Mathematical modeling of protein production with dynamics of resource and ribosomal population}
%%%%%%%%%%%%%%%%%%%%%%%%%%%%%%%%%%%%%%%%%%%%%%%%%%%%%%%%%%%%%%%%%%%%%%%%%%

We model a system of protein translation, which takes into account the production and maintenance of ribosomes and the dynamics of common resources (e.g. amino acids) used for protein synthesis. We distinguish ribosomal proteins, which are components and precursors of full functional ribosomes (denoted with $R$), and other types of proteins ($P$). In the model, $R$ captures both ribosomes and ribosomal proteins, assuming that ribosomes are instantly assembled from ribosomal proteins with perfect stochiometry, which is a common assumption for this type of model \cite{bremer1975parameters,chure2023optimal,kostinski2020ribosome,kostinski2021growth}. We identify the available resources with the variable $E$, with an external intake function $\alpha$ ($ \mu m^{-3}h^{-1}$). In the application of our model, $E$ represents certain type of amino acids (arginine/lysine) that are specifically recycled during starvation \cite{wyant2018nufip1}. However, our model and its analysis are more general and could be applied to other type of resources, including ATP with no recycling of resources involved for example \cite{basan2015overflow}. For simplicity, we assume the system to be in constant volume (for instance a non-proliferating cell), so that we can express both ribosomes $R$ and resource $E$ either in number of particles or concentration ($\mu m^{-3}$). The various mechanisms underlying the dynamics of $R$ and $E$ are shown in Figure \ref{fig:diagram_system}, with the following differential equations describing their evolution, as

\beq
\frac{dR}{dt} &=& \left[\gamma_R(E) V - U-\beta_R\right] R \label{eq:dRdt}\\
\frac{dE}{dt} &=& \alpha(t) + \left[- c_R \gamma_R(E) V - c_P \gamma_P (E)(1-V) + c_U U\right]R. \label{eq:dEdt}     
\eeq

%%%%%%%%%%%%%%%%%%%%%%%%%%%%%%%%%%%%%%%%%%%%%%%%%%%%%%%%%%%%%%%%%%%%%%%%%%
\begin{figure}[!ht]
	\centering
        \includegraphics[width=0.7\textwidth]{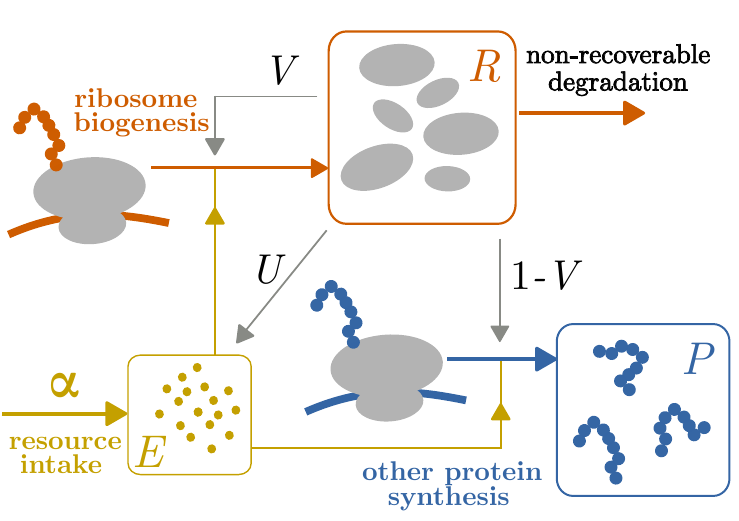}
	\caption{Mathematical model of protein production dynamics with resource and ribosomal population. A fraction $V$ of ribosomes (R) is allocated to produce other ribosomes (ribosome biogenesis), while the remaining fraction $1-V$ is responsible for the remaining protein synthesis (P). Biosynthesis rates depend on the amount of cellular resources (E) injected in the system (resource intake). Through the control $U$, the ribophagy process contributes to the pool E by recycling resources.}
\label{fig:diagram_system}
\end{figure}
%%%%%%%%%%%%%%%%%%%%%%%%%%%%%%%%%%%%%%%%%%%%%%%%%%%%%%%%%%%%%%%%%%%%%%%%%%

The dynamics of ribosomes $R$ (Eq.\eqref{eq:dRdt}) is governed by  mechanisms that account for their production and degradation.  
The production of ribosomes is mediated by $R$, at a global rate $\gamma_R(E) V R$, such that a fraction $0\leq V \leq 1$  of the ribosome population is allocated to it, at a resource-dependent rate function $\gamma_{R}(E)$ ($h^{-1}$). The function $ V $ (unitless) allows to shift the fraction of resources allocated for the production of ribosomes, accounting for the control of the ribosome biogenesis. It depends on nutrient conditions~\cite{scott2010interdependence, dai2020coupling, metzl2017principles} and implicitly includes the fraction of active ribosomes~\cite{dai_reduction_2017}. Similarly, non-ribosomal proteins are produced at a rate $\gamma_P(E) (1-V) R$, where $\gamma_{P}(E)$ ($h^{-1}$) is another positive rate function. Our model also distinguishes passive non recoverable degradation $\beta_R R$ at exponential rate $\beta_R$ ($h^{-1}$), and active degradation into recoverable resource  $UR$, where  $U$ ($h^{-1}$)  is a positive rate function that models the control of the population. 
This active degradation term, quantified by $c_U UR$ in Eq.\eqref{eq:dEdt}, where $c_U$ is a conversion factor, leads to the recycling of the  ribosome into resource for the system (as suggested in \cite{sinturel2017diurnal}). For simplicity, we will refer to the active degradation mechanism as \emph{ribophagy} \cite{nakatogawa2018spoon}. Similarly, $c_R$ models the amount of resources used per ribosome produced  (accounting for ribosomal proteins and proteins involved in the ribosomal assembly process), and $c_P$ for non-ribosomal proteins. The biological significance of $U$ and $V$ is further examined in Discussion.

%%%%%%%%%%%%%%%%%%%%%%%%%%%%%%%%%%%%%%%%%%%%%%%%%%%%%%%%%%%%%%%%%%%%%%%%%%
\subsection{Biophysical constraints and rate functions}
%%%%%%%%%%%%%%%%%%%%%%%%%%%%%%%%%%%%%%%%%%%%%%%%%%%%%%%%%%%%%%%%%%%%%%%%%%

For the system to produce physical solutions, we detail here the conditions satisfied by the parameters, that we will further use to study the model. First, mass conservation between 
 resources $E$ and  ribosomes $R$ (amount of resources available by recycling \a ribosome cannot exceed the amount used to create the ribosome) yields

\beq
\frac{c_U}{c_R}  \leq 1. \label{eq:mass_conservation} 
\eeq
For the system to produce non-trivial steady-state solutions in ribosomes $R$, the overall rate of degradation given by ($\beta_R + U$) is also less than the maximal creation rate, denoted by $\gamma_{R,\max}\coloneqq \sup_E\gamma_{R}(E)$, so

\beq
\beta_R + U < V \gamma_{R,\max}. \label{eq:constraint2}
\eeq
The  rate functions $\gamma_{i}(E)$ ($i=P,R$) are strictly increasing functions of the concentration of resources (e.g. amino acids), and should converge to a maximum constant because of chemical kinetic limits, whereas at low concentration, we can assume that they are linear. As a first approximation, we hence model the elongation with Michaelis-Menten kinetics \cite{dai_reduction_2017,klumpp_molecular_2013}, where $\gamma_{i,\max}$ ($ h^{-1}$) is the maximum production rate and $k_D$ ($ \mu m^{-3}$) is the Michaelis-Menten constant. We also assume that the same constant $k_D$ is shared for proteins $P$ and ribosomes $R$, as it reflects the average elongation speed over the global amino acid pool (for detailed derivation and interpretation, see Appendix \ref{appendix_5}), as

\beq \label{eq:gamma_function}
	\gamma_{i}(E) = \gamma_{i,\max}\frac{E}{E+k_D}, \qquad (i=P,R).
\eeq
Since we are interested in studying the mechanisms that lead the ribosome population to diurnally oscillate depending on feeding-fasting rhythms, we consider the following piecewice constant function with period $\tau$ ($h$) for the nutrient intake function 
\begin{equation}\label{eq:def_alpha}
    \alpha(t) = \left\{
    \begin{array}{ll}
        \alpha_{\max} & \mbox{if } 0\leq t<\frac{\tau}{2}\\
        \alpha_{\min} & \mbox{if }  \frac{\tau}{2}\leq t<\tau.
    \end{array}
\right.
\end{equation}

%%%%%%%%%%%%%%%%%%%%%%%%%%%%%%%%%%%%%%%%%%%%%%%%%%%%%%%%%%%%%%%%%%%%%%%%%%
\subsection{Optimization of gene expression}
%%%%%%%%%%%%%%%%%%%%%%%%%%%%%%%%%%%%%%%%%%%%%%%%%%%%%%%%%%%%%%%%%%%%%%%%%%

Using the model, our goal is to study how the aforementioned control mechanisms yield solutions for the dynamical system, that also optimize for some objective function associated with protein production.
We turn our attention to the regulation of genes that control essential cellular functions, and assume that their expression should be stably maintained as much as possible under varying nutrient conditions \cite{fraser2004noise,joshi_what_2022}. We thus consider the following optimisation problem, in which the ribosome population is tuned over time
to enable high stable protein production. The production rate of non-ribosomal proteins in the model is expressed as
    \beq\label{eq:rho}
    \rho=\gamma_P (E) (1-V)R,
    \eeq 
which depends on the dynamics of the system, and thus also varies as a function of $U$ and $V$. We further assume that the control functions $U,V$ take their value in $\mathcal{D} = [U_{\max},U_{\min}]\times  [V_{\max},V_{\min}]$ respectively, with the condition \eqref{eq:constraint2} for non trivial steady states being satisfied, so that $\beta_R+U_{\max}<\gamma_{R,\max} V_{\min}$. To enable a maximal and constant non-ribosomal production in the model, we thus have to find control functions $\hat{U}$ and $\hat{V}$ that are solutions of the following optimization problem
\beq\label{eq:optimal_formulation}
\rho(\hat{U},\hat{V})=\max\limits_{U,V}\{\rho(U,V) \ | \ \rho(U,V)(t) \textrm{ is constant, and } \forall t , \ (U(t),V(t)) \in \mathcal{D}\}\,.
\eeq

%%%%%%%%%%%%%%%%%%%%%%%%%%%%%%%%%%%%%%%%%%%%%%%%%%%%%%%%%%%%%%%%%%%%%%%%%%
\section{Results}
%%%%%%%%%%%%%%%%%%%%%%%%%%%%%%%%%%%%%%%%%%%%%%%%%%%%%%%%%%%%%%%%%%%%%%%%%%

%%%%%%%%%%%%%%%%%%%%%%%%%%%%%%%%%%%%%%%%%%%%%%%%%%%%%%%%%%%%%%%%%%%%%%%%%%
\subsection{Global stability analysis}\label{theor_analy}
%%%%%%%%%%%%%%%%%%%%%%%%%%%%%%%%%%%%%%%%%%%%%%%%%%%%%%%%%%%%%%%%%%%%%%%%%%

To solve the optimization problem \ref{eq:optimal_formulation}, we first study the properties of the system for constant input $\alpha$ and control functions $U, V$. In this case, the system forms an autonomous dynamical system, which we can characterize in terms of equilibrium points and stability. In other words, we can evaluate, for a given amount of source nutrient, what levels of nutrient and ribosomes available can be observed at steady state.

\begin{proposition}\label{prop:stability} Let us assume $\alpha$ constant and the rate functions $\gamma_{i}(E)$ ($i=P,R$) satisfy $\gamma_i(0)=0$,  $\gamma_{i}$ is $\mathcal{C}^1$ with positive derivative $\gamma_i'>0$, and $\lim_{E\rightarrow \infty} \gamma_i(E) = \gamma_{i,\max} < \infty$.
Then, the solutions of the autonomous system \eqref{eq:dRdt}-\eqref{eq:dEdt} starting in $\R_{>0}^2$ remain strictly positive and converge to a globally asymptotically stable equilibrium point 
$(R^*,E^*) \in \R_{>0}^2$, given by
\beq
R^* &=& \frac{\alpha}{c_R(\beta_r+U)+c_P\gamma_P(E^*)(1-V)-c_U U}\label{eq:R*0}\\
E^* &=& \gamma_R^{-1}\left(\frac{\beta_r+U}{V}\right) \label{eq:E*0}.
\eeq
Moreover, all solutions initialized in $\R_{>0}^2$ oscillate, in the sense that $R(t)-R^*$ changes sign infinitely often as $t\to\infty$, if and only if $\Delta<0$, where
\beq
\Delta=(c_R\gamma_R'(E^*)V+c_P\gamma_P'(E^*)(1-V))^2R^{*2}-4\alpha \gamma_R'(E^*)V.
\label{eq:Delta}
\eeq
\end{proposition}

\noindent A detailed proof of this result is provided in Appendix \ref{sec:appendix_stability}, where we apply the Hartman-Grobman theorem \cite{guckenheimer_nonlinear_1983} to study the local stability and behaviour of solutions near the equilibrium point, before introducing the Poincaré map to prove that it is globally asymptotically stable. The proof also requires assumptions of regularity and monotony of the rate functions $\gamma_R(E)$ and $\gamma_P(E)$ (listed as (ii) in section \ref{section:model}). Note that when $U\ge  V \gamma_{R,\max}-\beta_R$ (see Eq. \eqref{eq:constraint2}), the ribosomal population goes to extinction (Appendix \ref{sec:appendix_stability}). Figure \ref{fig:Pportrait_Tseries} illustrates the second part of the proposition, with two examples, one with  $\Delta<$ 0  (Figure \ref{fig:Pportrait_Tseries_osc} ) that leads to oscillations, and  one with $\Delta>$0 (Figure \ref{fig:Pportrait_Tseries_non_osc}) that does not.

%%%%%%%%%%%%%%%%%%%%%%%%%%%%%%%%%%%%%%%%%%%%%%%%%%%%%%%%%%%%%%%%%%%%%%%%%%
\begin{figure}[!ht]
    \centering
    \begin{subfigure}[b]{0.48\textwidth}
        \includegraphics[width=\textwidth]{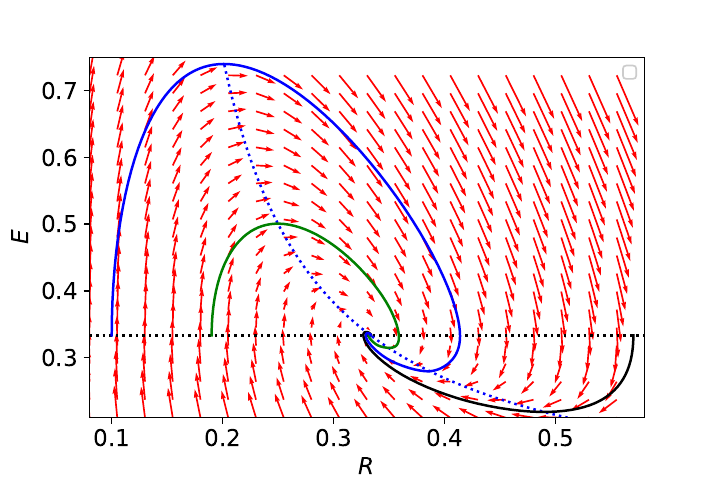}
        \includegraphics[width=\textwidth]{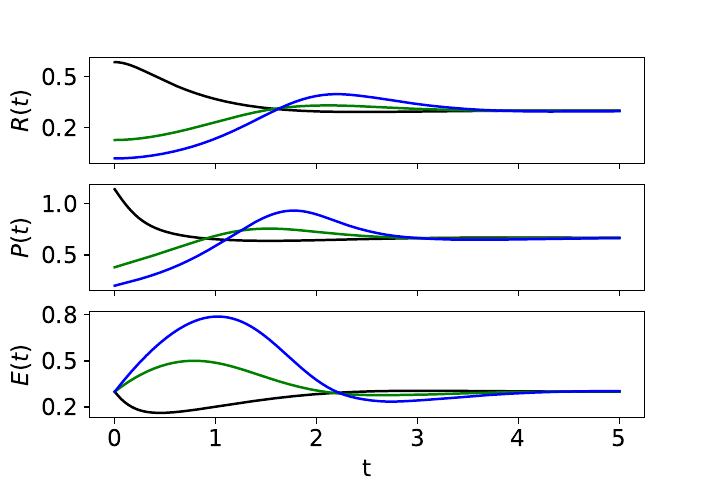}
        \caption{ $\Delta>0$, no oscillations}
        \label{fig:Pportrait_Tseries_non_osc}
    \end{subfigure}
    \unskip\ \vrule\ 
    \centering
    \begin{subfigure}[b]{0.48\textwidth}
        \includegraphics[width=\textwidth]{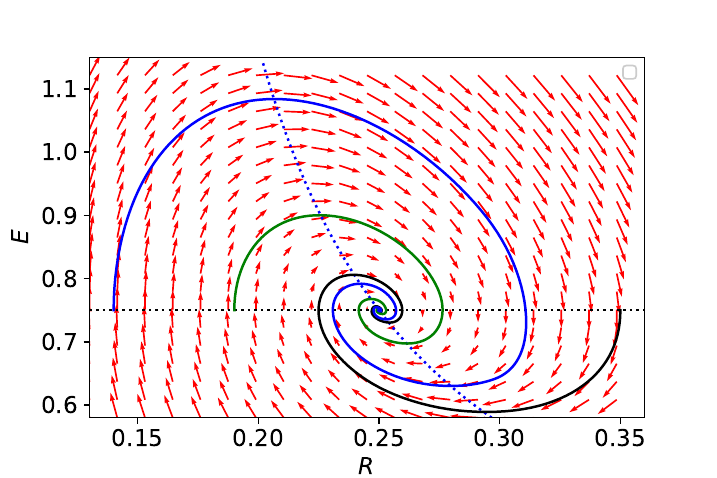}
        \includegraphics[width=\textwidth]{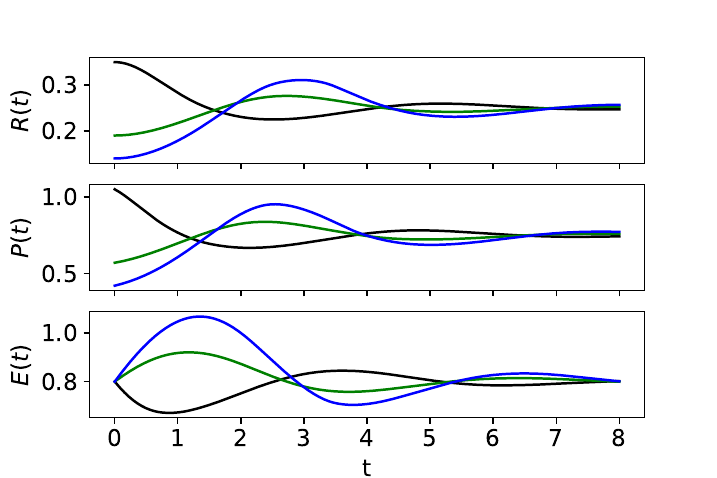}
        \caption{ $\Delta<0$, oscillations}
        \label{fig:Pportrait_Tseries_osc}
    \end{subfigure}
    \medskip
     \hrule 
    \medskip
    \begin{subfigure}[t]{0.48\textwidth}
        \includegraphics[width=\textwidth]{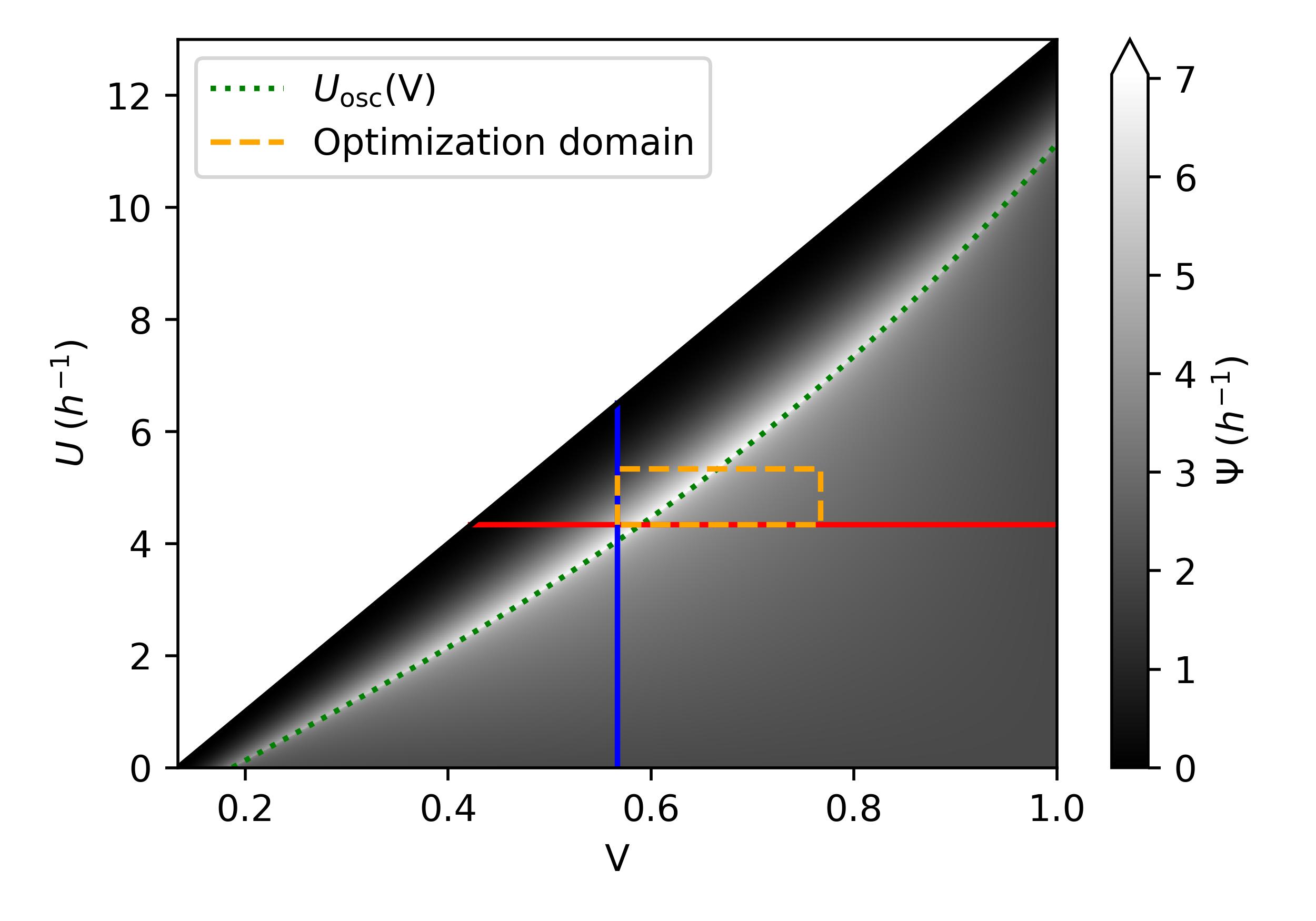}
        \caption{Heatmap of $\Psi$ as a function of $U$ and $V$}
        \label{heatmap_phi}
    \end{subfigure}
    \centering
    \begin{subfigure}[t]{0.48\textwidth}
        \includegraphics[width=\textwidth]{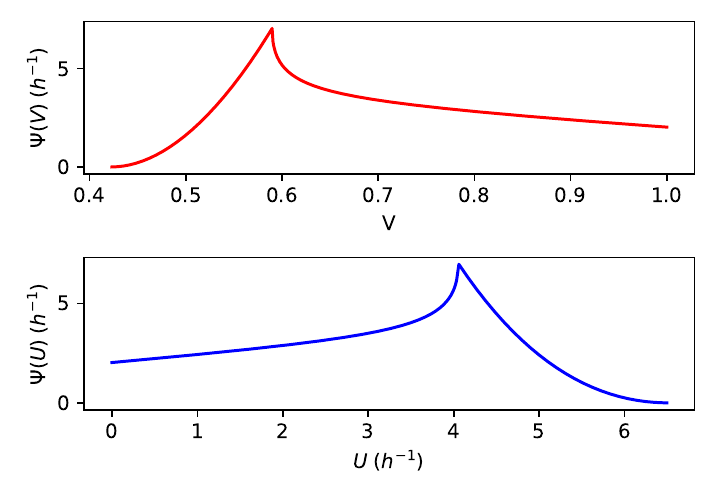}
        \caption{Convergence speed as a function of one parameter, with $U=4.3\;h^{-1},\;V=0.57$}
        \label{heatmap_Phase_UV}
    \end{subfigure}

    \caption{Phase portrait, time trajectories and convergence speed of the system, with resources $E$, ribosome population $R$ an protein production $P$.  Nullclines are represented in dotted black and blue lines (see also Appendix \ref{sec:appendix_stability}), with some trajectories in black, green and blue in the  non oscillatory case \ref{fig:Pportrait_Tseries_non_osc}, and in the oscillatory case \ref{fig:Pportrait_Tseries_osc}. Figure \ref{heatmap_phi} represents the asymptotic convergence speed of the system $\Psi$ as a function of the ribophagy rate $U$ and the ribogenesis rate $V$, for nutrient intake $\alpha= 2\cdot 10^8\; h^{-1}\mu m^{-3}$. The dotted box represents an example of the optimization domain. Figure \ref{heatmap_Phase_UV} shows a vertical and horizontal slice from Figure \ref{heatmap_phi}.}
    \label{fig:Pportrait_Tseries}
\end{figure}
%%%%%%%%%%%%%%%%%%%%%%%%%%%%%%%%%%%%%%%%%%%%%%%%%%%%%%%%%%%%%%%%%%%%%%%%%%

%%%%%%%%%%%%%%%%%%%%%%%%%%%%%%%%%%%%%%%%%%%%%%%%%%%%%%%%%%%%%%%%%%%%%%%%%%
\subsection{Quantification of the convergence to equilibrium}\label{sec:results_psi}
%%%%%%%%%%%%%%%%%%%%%%%%%%%%%%%%%%%%%%%%%%%%%%%%%%%%%%%%%%%%%%%%%%%%%%%%%%

In this section, we quantify how quickly steady state levels can be reached, since in the next section when optimizing protein production, we will assume that convergence to the steady state is fast. The asymptotic convergence rate describes how fast the biological system react to small perturbations and will enable us to describe the system only with its steady states. We now assume that the elongation speed is given by a Michaelis Menten dynamics \eqref{eq:gamma_function}, and that the controls $(U,V)$ satisfy assumptions of section \ref{section:model} 
Combining Eq.\eqref{eq:gamma_function} with Proposition \ref{prop:stability}, we found that the steady state is
\beq
R^* &=& \frac{\alpha}{c_R(\beta_r+U)+\frac{c_P\gamma_{P,\max}(1-V)(\beta_R+U)}{\gamma_{R,\max}V}-c_U U} \label{R_star} ,\\
E^* &=& \frac{k_D}{\frac{V \gamma_{R,\max}}{(\beta_R+U)} -1}. \label{E_star}
\eeq
Similarly, setting $\Delta=0$ in Eq.\eqref{eq:Delta} and using Eq.\eqref{eq:gamma_function} yields critical values of nutrient intake $\alpha_{osc}$ and ribophagy rate $U_{osc}$ that both depend on $V$, to guarantee the existence of oscillations around the steady state $(R^*,E^*)$, as detailed in Proposition \ref{prop:oscillation} of Appendix \ref{appendix_2}. In Proposition \ref{prop:Psi_variation}, we then determine the behaviour of the convergence rate of the system towards equilibrium, defined by $\Psi(U,V):=-\real(\lambda)$, where $\lambda$ is the eigenvalue  of the Jacobian matrix with the largest real part, for the system at steady state $(R^*,E^*)$. In particular, we can derive from this proposition a lower bound for the asymptotic convergence rate to this steady state (stated in lemma \ref{lemma_lower_bound} Appendix \ref{appendix_2}), which will be useful for solving the optimal control problem \eqref{eq:optimal_formulation}, as
\beq
\label{eq:lower_bound}
\Psi(U,V) \geq \min \left( \Psi(U_{\min},V_{\max}),\Psi(U_{\max},V_{\min})\right).
\eeq
In Figure \ref{heatmap_phi}, we plot the values of the convergence rate $\Psi(U,V)$ in a heatmap, that illustrates that the lower bound in Eq. \eqref{eq:lower_bound} is indeed located either at $\Psi(U_{\min},V_{\max})$ or $\Psi(U_{\max},V_{\min})$, and shown in Figure \ref{heatmap_Phase_UV}. Figure \ref{heatmap_phi} also contains a line that defines an oscillation threshold for the ribophagy control $U_{\osc}(V)$, such that there is an infinite number of oscillations near the steady state for $U > U_{\osc} $, and a finite number else. We also note that the convergence rate $\Psi$ is not defined when constraint \ref{eq:constraint2} is violated (shown by the blank region in Figure \ref{heatmap_phi}), and goes to 0 as $\gamma_{R,\max}V-(\beta_R+U)\to 0^+$ (so the convergence time to the steady state goes to infinity). Knowing the rate of convergence $\Psi$ of the system can help with inferring the values of the controls $U,V$: if the convergence rate is measured experimentally (using the ribosomal population over time for example), a domain of admissible parameters can be estimated using the analytical formula of the convergence speed $\Psi$, or the values of the ribosome population and resources at steady state.

%%%%%%%%%%%%%%%%%%%%%%%%%%%%%%%%%%%%%%%%%%%%%%%%%%%%%%%%%%%%%%%%%%%%%%%%%%
\subsection{Optimal control under quasi-static approximation}\label{sec:results_oc}
%%%%%%%%%%%%%%%%%%%%%%%%%%%%%%%%%%%%%%%%%%%%%%%%%%%%%%%%%%%%%%%%%%%%%%%%%%
In light of the theoretical results described in the previous sections, we study how the ribophagy rate $U$ and the fraction allocating ribogenesis $V$ can allow for a high and stable protein production $\rho$. More precisely, we study the solutions associated with the optimal control problem defined in Eq.\eqref{eq:optimal_formulation}, where the intake function $\alpha(t)$ periodically switches between two values (Eq.\eqref{eq:def_alpha}). In this case, the convergence rate function $\Psi(U,V)$ depends on $\alpha$, and we assume here that this rate is large compared with the switching frequency ($\frac{2}{\tau}$) of the intake function $\alpha(t)$, that is

\begin{equation}\label{eq:Hypothesis_char_time}
  \forall (U,V,\alpha) \in \left\{U_{\min} , U_{\max} \right\}\times  \left\{ V_{\min} , V_{\max} \right\}\times \left\{ \alpha_{\min} , \alpha_{\max} \right\}, \qquad  \Psi_{\alpha}(U,V) \gg \frac{2}{\tau}.
\end{equation}
If this condition is satisfied, then using lower bounds on the convergence rate  Eq.\eqref{eq:lower_bound} makes $\Psi_{\alpha}(U,V) \gg \frac{2}{\tau}$ true for all $(U,V) \in\mathcal{D}$. In other words, all the solutions associated with control values defined in the set of admissible parameters $\mathcal{D}$ converge quickly to steady state $(R^*,E^*)$ compared to the period of $\alpha(t)$. This condition is necessary for the system to be quasi-static, in the sense that the characteristic time to converge to a new equilibrium is short relative to the time between changes in the value of the intake function $\alpha$. Note that this condition is not sufficient, so the approximation needs some further (numerical) validation.
The conditions for existence and solving the optimal control can then be formulated as 
\begin{proposition}\label{prop:cst_ctrl}
Let $\rho^*(U,V, \alpha)$ be the value of the production rate function $\rho$ (Eq. \eqref{eq:rho}) at the steady state $(R^*,E^*)$ associated with $(U,V, \alpha)$ (Eq. \eqref{eq:R*0}-\eqref{eq:E*0}), and assume that the quasi-static approximation holds (Eq. \eqref{eq:Hypothesis_char_time}). 
Then the optimal control problem Eq. \eqref{eq:optimal_formulation} admits a solution if and only if the equation 

\begin{equation}\label{eq:usvs}
\rho^*(U,V,\alpha_{\max}) = \rho^*(U_{\max}, V_{\min},\alpha_{\min})
\end{equation}
admits a solution $(U_s,V_s)$. In this case,
\begin{equation*}
    (\hat{U},\hat{V})(t) = \left\{
    \begin{array}{ll}
        (U_{\max},V_{\min})  & \mbox{if  } 0\leq t<\frac{\tau}{2}\\
        (U_s,V_s) & \mbox{if } \frac{\tau}{2} \leq t< \tau
    \end{array}
\right.
\end{equation*}
is optimal for Eq. \eqref{eq:optimal_formulation}, and solutions of Eq. \eqref{eq:usvs} define a 1D-manifold $\Gamma=[U_{\min},U_{\max}]\times[V_{\min},V_{\max}]\cap\{(\Tilde{F}_{\alpha_{\max}}(V),V), V\in[0,1]\}$, such that $\Tilde{F}_{\alpha_{\max}}$ is an increasing function of $V$.

\end{proposition}

We study the variation of the objective function $\rho^*$ in Appendix \ref{appendix_3} and prove Proposition \ref{prop:cst_ctrl}, by finding the values of $U$ and $V$ that maximize the protein production rate $\rho^*$ (Eq. \eqref{eq:rho}) at low intake (when $0\leq t \leq \frac{\tau}{2}$), while matching with the protein production rate at high intake (when $\frac{\tau}{2}\leq t \leq \tau$). Also note that the function $\tilde{F}_{\alpha_{\max}}$ in Proposition \ref{prop:cst_ctrl} can be exactly expressed (see Eq. \eqref{eq:tildeF}, Appendix \ref{appendix_3}), and that the protein production is constant along $\Gamma$, so another criterion (e.g. maximizing the convergence rate to equilibrium) is required to obtain a unique optimal solution.

%%%%%%%%%%%%%%%%%%%%%%%%%%%%%%%%%%%%%%%%%%%%%%%%%%%%%%%%%%%%%%%%%%%%%%%%%%
\subsection{Application to eukaryotic liver cells}\label{section:numerical_application}
%%%%%%%%%%%%%%%%%%%%%%%%%%%%%%%%%%%%%%%%%%%%%%%%%%%%%%%%%%%%%%%%%%%%%%%%%%

In this section, we directly apply our results to the case of eukaryotic liver cells that motivated our analysis \cite{sinturel2017diurnal}. To do so, we use a biological set of parameters detailed in Appendix \ref{appendix_5}, with the only arbitrary parameters being the range of the control by ribophagy $U_{\min}$ and $U_{\max}$ (that is, to our knowledge, not accessible from the existing literature). The other parameters are derived, directly of indirectly, from experimental data. 
To verify if the quasi-static approximation holds under our set of parameters, we find the minimum of the convergence rate $\Psi$ over the admissible parameters $\mathcal{D}$ using Eq. \eqref{eq:lower_bound}, and we check that Eq. \eqref{eq:Hypothesis_char_time} is satisfied, with $
\min \Psi =1.0183 \;h^{-1}\;\gg\; \frac{2}{\tau}=0.0833\;h^{-1}$ (we also validated the quasi-static approximation numerically, by computing the convergence rate over the set of parameters $\mathcal{D}$).

%%%%%%%%%%%%%%%%%%%%%%%%%%%%%%%%%%%%%%%%%%%%%%%%%%%%%%%%%%%%%%%%%%%%%%%%%%
\begin{figure}[!ht]
    \includegraphics[width=\textwidth]{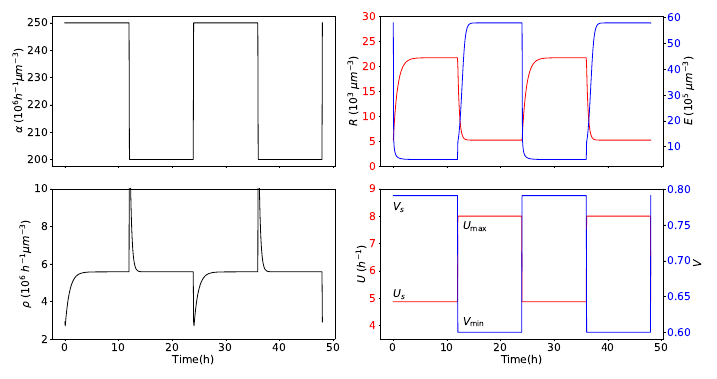}
    \caption{Example of an optimal control solution obtained from solving Eq. \eqref{eq:usvs} (Proposition \ref{prop:cst_ctrl}) with maximal constant protein production and maximal convergence rate (to enforce the quasi-static approximation, see section \ref{section:numerical_application}). Plotted variables are the nutrient intake $\alpha$, ribosome concentration $R$ and a resource concentration $E$, with the system producing proteins at a rate $\rho$ under the ribophagy $U$ and ribogenesis $V$ controls. Optimization was performed with the Newton method of the \texttt{scipy.optimize.minimize} python function.We used parameters from Table \ref{table:parameters}, as described in Appendix \ref{appendix_5}. Simulations were performed using the \texttt{solve\_ivp} function of the python \texttt{scipy} package \cite{2020SciPy-NMeth}.}
    \label{Heatmaps_cste_ctrl}
\end{figure}
%%%%%%%%%%%%%%%%%%%%%%%%%%%%%%%%%%%%%%%%%%%%%%%%%%%%%%%%%%%%%%%%%%%%%%%%%%

 In Figure \ref{Heatmaps_cste_ctrl}, we illustrate the optimal control solutions, with a specific choice for the ribophagy and ribogenesis controls $(U,V)$ such that the convergence rate is maximized. In particular, the plots indicate that during rich-nutrient periods, ribosome population $R$  and control by ribogenesis $V$ are high, whereas the ribophagy $U$ and resources $E$ are low, and they all switch to opposite values when the system is under starvation. Note that the sharp peaks and troughs of proteins produced when switching happens is due to the discontinuity in the nutrient intake and control functions, and can be softened by smoothing these functions. Overall, our model is able to qualitatively reproduce the observations found in the mice liver, as  ribosome population level significantly increase during feeding periods, while polyadenylation and degradation of the rRNAs of incomplete ribosomal subunits lead to their decrease during starvation \cite{sinturel2017diurnal}. We also note that regardless of the control values for $U,V$ chosen on a 1D manifold that yields the optimal solutions (see Proposition \ref{prop:cst_ctrl}), we always obtain the same qualitative behaviour.

%%%%%%%%%%%%%%%%%%%%%%%%%%%%%%%%%%%%%%%%%%%%%%%%%%%%%%%%%%%%%%%%%%%%%%%%%%
\subsection{Constraints on $\alpha, U$ and $V$ for the existence of optimal control solutions}\label{sec:constraints}
%%%%%%%%%%%%%%%%%%%%%%%%%%%%%%%%%%%%%%%%%%%%%%%%%%%%%%%%%%%%%%%%%%%%%%%%%%

While the results of the previous section were obtained with an arbitrary  choice of domain values for the ribophagy rate $U_{\min}$ and $U_{\max}$, we can further study how the values of the controls are constrained to guarantee the existence of the optimal control solutions. First, we assume that the high nutrient condition is set, with the maximum resource intake $\alpha_{\max}$ given. Then for   given bounds on the ribogenesis fraction $V_{\min}$ and $V_{\max}$  (like in the previous section), we study the values of $U_{\min}$ and $U_{\max}$, such that the system can produce constant protein production over time. If it can, from proposition \ref{prop:rho_star} (Appendix \ref{appendix_3}), we know that  

\begin{equation}\label{eq:constraints_rho}
    \rho^*(\alpha_{\min}, U_{\max}, V_{\min}) \geq \rho^*(\alpha_{\min}, U, V_{\min}) = \rho^*(\alpha_{\max}, U_{\min}, V_{\max}). 
\end{equation}
We show in Appendix \ref{appendix_4} that this equation can produce a solution $U$ only if 
\begin{equation}\label{eq:alphamin_lb}
    0<q_1\leq\frac{\alpha_{\min}}{\alpha_{\max}} \leq q_2,
\end{equation}
where $q_1,q_2$ are constant threshold values that are analytically known (see Appendix \ref{appendix_4}), and can be interpreted as follows: The first threshold value $q_1$ indicates when starvation becomes too strong for the system to keep maintaining constant protein production. The second threshold $q_2$ indicates when the combined effects of ribophagy and ribogenesis cannot sufficiently reduce the protein production in high nutrient conditions.

Within this valid range of values for the minimum intake $\alpha_{min}$, the bounds on the ribophagy rate $(U_{\min},U_{\max})$ producing an optimal solution $U_s$, as defined in proposition \ref{prop:cst_ctrl}, is then given by a closed domain derived explicitly in Appendix \ref{appendix_4}. Figure \ref{fig:3d_plot1} a illustrates the domains obtained for different values of $\alpha_{\min}$, showing the size of the domain of valid parameters converging to 0 as $\alpha_{\min}$ gets close to the bounds given in equation \eqref{eq:alphamin_lb}. 

%%%%%%%%%%%%%%%%%%%%%%%%%%%%%%%%%%%%%%%%%%%%%%%%%%%%%%%%%%%%%%%%%%%%%%%%%%
\begin{figure}[!ht]
    \centering
    \begin{subfigure}[t]{.48\textwidth}
        \includegraphics[width=\textwidth]{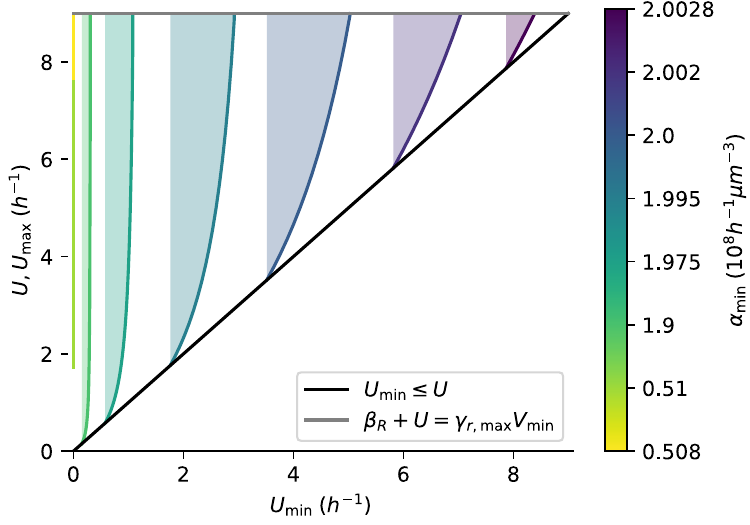}
        \caption{}
        \label{fig:3d_plot1}
    \end{subfigure}
    \unskip\ \vrule\ 
    \centering
    \begin{subfigure}[t]{0.48\textwidth}
        \includegraphics[width=\textwidth]{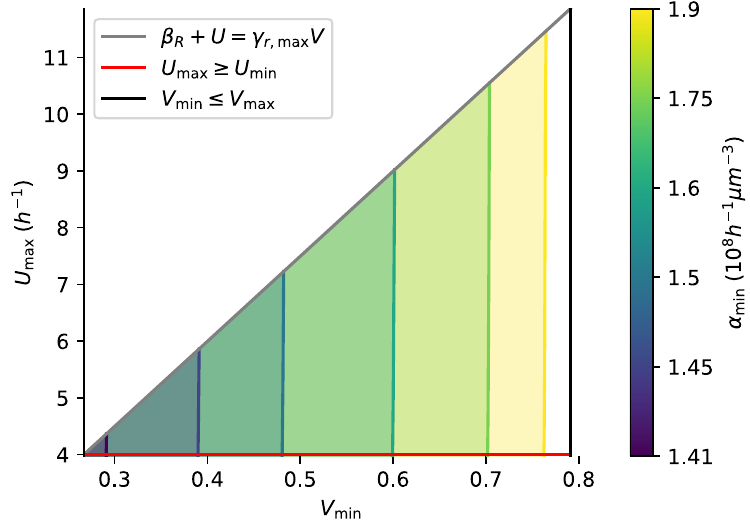}
        \caption{}
        \label{fig:3d_plot2}
    \end{subfigure}
    \caption{Representation of regions of admissible parameters in two different settings, as a function of the minimum resource intake $\alpha_{\min}$. The grey, red and black lines represent different inequalities as a boundary of the parameter domain. The other colored curves represent, for different values of $\alpha_{\min}$ displayed on the color bar, the value of controls such that the protein production stays constant during low intake period (starvation). In Fig.\ref{fig:3d_plot1}, the  ribogenesis control value range $(V_{\min},V_{\max})$ is fixed and the value range of ribophagy $(U_{\min},U_{\max})$ is studied, whereas in Fig.\ref{fig:3d_plot2}, the protein production is fixed ($\rho_0$), and the values of the controls $(U_{\max},V_{\min})$ are studied at low intake period (starvation). The maximal resource intake $\alpha_{\max}$ was set to $2.5\cdot 10^{-8} h^{-1}\mu m^{-3}$.} 
    \label{fig:3d_plot}
\end{figure}
%%%%%%%%%%%%%%%%%%%%%%%%%%%%%%%%%%%%%%%%%%%%%%%%%%%%%%%%%%%%%%%%%%%%%%%%%%
 
 Alternatively, we can also constrain the parameters by assuming a given value $\rho_0$ for the protein production at high resource intake $\alpha=\alpha_{\max}$, and look for the bounds on ribophagy and ribogenesis controls  $(U_{\max},V_{\min})$ such that there exists an optimal control solution $(U_s,V_s)$ during low intake periods (see Figure \ref{Heatmaps_cste_ctrl}). In this case, we can derive  a constant value $q_3>0$, such that we have the following lower bound on the low resource intake value (see Appendix \ref{appendix_4}):
\begin{equation}\label{eq:domain_UV}
    \alpha_{\min}\geq\rho_{0}q_3.
\end{equation}

In Figure \ref{fig:3d_plot2}, we again illustrate the domains associated with equation \eqref{eq:domain_UV}, for different values of parameter $\alpha_{\min}$, showing that the domain shrinks to 0 as $\alpha_{\min}$ decreases. The different bounds can be used to predict if the system can produce proteins at a constant rate or not. Conversely, if we know how the system evolves under a certain resource intake function $\alpha$, these bounds can be used to further analytically extract other parameters of interest, such as the ribophagy rate $U$ or the ribogenesis ratio $V$.
%%%%%%%%%%%%%%%%%%%%%%%%%%%%%%%%%%%%%%%%%%%%%%%%%%%%%%%%%%%%%%%%%%%%%%%%%%
\section{Discussion}
%%%%%%%%%%%%%%%%%%%%%%%%%%%%%%%%%%%%%%%%%%%%%%%%%%%%%%%%%%%%%%%%%%%%%%%%%%

In this paper, we introduced a coarse grain system that models the dynamics of the ribosome population, with two control mechanisms that account for ribosome biogenesis and recycling. Motivated by experimental observations of oscillations linked to periods of high and low nutrients \cite{sinturel2017diurnal}, we formulated an optimal control problem. We solved this problem after proving results on the stability and convergence speed of the system, as it allowed us to find analytical solutions under a quasi-static approximation \cite{wachsmuth2012optimal}. We showed that for a set of biophysical parameters and specific input and control values, the system can produce optimal solutions in qualitative agreement with experimental observations. Note that since these parameters were taken from different species, and since these input and control values remain unknown, we cannot obtain a precise quantitative comparison with experimental data. Yet, our study shows that the main qualitative behaviour of the optimal solutions is robust to change in the parameters, as long as the quasi-static approximation holds.

Optimal control theory and methods have been used in various biological applications, including for optimizing drug scheduling \cite{swan1981optimal}, pathogens treatment design \cite{ewald2020trends}, biosynthesis in bioreactors \cite{lunz2023optimal} or cellular growth \cite{giordano2016dynamical}. In the context of gene expression, other studies also investigated how the behaviour of the system can be expressed as a solution of an optimal control problem, with a cost associated with some metabolic pathways \cite{klipp2002prediction} (e.g. the lac system of \textit{E.coli} \cite{dekel2005optimality}), that also depends on the environment \cite{poelwijk2011optimality, pavlov2013optimal,tsiantis2020using}.  While the mathematical solutions derived from these studies are time-dependent, our study models optimal controls that in contrast depend only on the state of the system, as the underlying biochemical reactions are sensitive to the amount of nutrient in the system \cite{wyant2018nufip1,tsokanos2016eif,kawamata_zinc_2017}. In addition, numerical approaches that are commonly used for solving optimal control problems in systems biology  \cite{tsiantis2020using,sharp2021implementation} were not adequate in our study, as we sought optimal solutions defined over a range of unknown parameters. By taking advantage of the relative simplicity of our model and using a steady state approximation similar to other studies in mechanics \cite{wachsmuth2012optimal,gao2021quasi}, we could explicitly solve the problem analytically and find optimal solutions that depend on the control and input parameters. To our knowledge, our study is hence the first to investigate the dynamics of ribosomes as solutions of an optimal control problem that includes ribophagy. The importance of ribosomes with regard to the production rate was previously investigated using a mathematical model, without taking into account the global dynamics of ribosomal population \cite{sabi2019modelling,katz2022translation,zarai2017controllability}. The impact on growth and metabolism was also more widely investigated \cite{scott2023shaping}, to derive various growth laws under hypothesis of growth-rate maximization \cite{scott_emergence_2014,scott2010interdependence}. Those laws have also been compared with numerical computation of optimal control associated with coarse-grained models of the cell  \cite{giordano2016dynamical}. While these previous studies specifically focused on studying the conditions for exponential growth of bacteria, our model relates with more complex eukaryotic systems that produce periodic patterns for the ribosome population, with the existence of regulatory pathways (e.g. ribophagy) that can tune the ribosome population.

In light of our findings, it would be interesting to study how to relate the control variables with their underlying biochemical pathways, which involve for example the Mechanistic Target of Rapamycin Complex 1 (mTORC1), that downregulates ribosome protein expression  when inactivated \cite{rosario2020mtorc1}, and can be triggered by amino acid starvation \cite{tsokanos2016eif}, as well as the interaction of lysosome and mTORC1 \cite{demetriades2014regulation}. For example, one could measure the impact of inhibition of mTOR pathways on translation \cite{rosario2020mtorc1}, to obtain a range of values for the ribogenesis control $V$ in our model. Another example of pathways induced by starvation is the inactivation of the TORC1 complex in yeast through zinc protein (such as the ribosomal protein Rpl37) that induces autophagy \cite{kawamata_zinc_2017}. More generally, the ribophagy control $U$, and its associated values ($U_{\min},U_{s}, U_{\max}$ etc.) models the rate of  ribosome active degradation under these molecular pathways \cite{nakatogawa2018spoon}, and could thus be potentially derived from more modeling at the molecular level. The ribophagy control $V$ represents the ratio between ribosomes allocated to ribosomal protein production and other essential proteins, which could be evaluated by comparing translation efficiencies across genes \cite{dao2018impact}. Alternatively, using the analytical description of the system as steady state, we can in principle indirectly estimate the values of both controls $U$ and $V$  from measurements of the convergence time  or response of the system to a change in nutrients (as specified in sections \ref{sec:results_psi} and \ref{sec:constraints}, respectively).

We note that while most of our biophysical parameters were derived from the literature (see Table \ref{table:parameters}), they are actually taken from different sources and organisms, and as such could require some correction to fit specific systems. We are in particular not aware of a quantification of the minimum and maximum ribophagy rates $U_{\min}$ and $U_{\max}$, as the literature is at its early stage of qualitative description and comprehension of ribophagy pathways and their properties \cite{wyant2018nufip1,nakatogawa2018spoon,kawamata_zinc_2017,Meiyan2018Ribophagy}. For example, it is unknown to what extent ribophagy and autophagy pathways are selective \cite{an_ribosome_2020}, and circadian rythms could also have an impact \cite{sinturel2017diurnal}, which we did not take into account in our model. 

There are also several assumptions, limitations and potential developments from our study. 
We first assumed that the volume of the system is fixed, while it can significantly vary in real systems (e.g. 20-30\% in \cite{sinturel2017diurnal}). While taking into account these variations would affect the dynamics of the system, it would actually not qualitatively impact the steady state and our theoretical results. More importantly, our simplified model reduces to two homogeneous pools of non-ribosomal and ribosomal proteins produced by the system, that neglect variations across genes, and would require some further compartmentalization to study the system at a finer scale. We also did not include noise in our model, as we considered large averaged populations, and focused our analysis on highly expressed genes with a low level of noise \cite{fraser2004noise}. In principle, our theoretical results can be extended for small enough fluctuations, by using methods for stochastic stability \cite{khasminskii2011stochastic}. Another interesting refinement of the model would be to take into account inactivation and activation mechanisms of the ribosome instead of having only passive degradation \cite{buskirk2017ribosome}, or include transcription dynamics \cite{chen1999modeling} (we restricted our model to translation, since it was cited as the predominant limiting factor in the experiments we modeled \cite{sinturel2017diurnal}). We note that such a refinement would imply working on dynamical systems of higher dimensions. In principle, separating time scales \cite{witelski2015methods} can still be used to solve optimal control with more dynamic variables (e.g. mRNA level for transcription, or different types of proteins). However, the global stability analysis would be more challenging, as we cannot leverage specific results for 2D systems anymore (e.g the Bendixson-Dulac theorem in Appendix A). 
It would also be interesting to extend the mathematical framework using delayed differential equations \cite{roussel1996use}, to study the potential impact of the dynamics of ribosomal assembly \cite{kressler_driving_2010} on the stability and optimal control solutions. The delay would represent the characteristic time of assembly of a ribosome. For a 2D stable dynamical system, we can prove that introducing a
small enough delay does not change local
stability (Proposition 1 in \cite{nelson2013dynamical}). Finally, our model and results could serve as a basis for studying the impact of ribosome control at both the macroscale level, with distinct populations of cells sharing resources, and the microscale of genes that share the pool of ribosomes to produce proteins.

%%%%%%%%%%%%%%%%%%%%%%%%%%%%%%%%%%%%%%%%%%%%%%%%%%%%%%%%%%%%%%%%%%%%%%%%%%
\paragraph{Acknowledgments:} This research is supported by a NSERC Discovery Grant PG 22R3468, a PIMS-CNRS student mobility fellowship, and the France Canada Research Fund (FFCR). This research was supported in part through computational resources and services provided by Advanced Research Computing at the University of British Columbia. Luca Ciandrini is supported by the French National Research Agency (REF: ANR-21-CE45-0009) and by the 
Institut Universitaire de France (IUF).
%%%%%%%%%%%%%%%%%%%%%%%%%%%%%%%%%%%%%%%%%%%%%%%%%%%%%%%%%%%%%%%%%%%%%%%%%%

\bibliographystyle{plos2015}
\bibliography{bibliography}

\begin{thebibliography}{10}

\bibitem{li2014quantifying}
Li GW, Burkhardt D, Gross C, Weissman JS.
\newblock Quantifying absolute protein synthesis rates reveals principles
  underlying allocation of cellular resources.
\newblock Cell. 2014;157(3):624--635.

\bibitem{bosdriesz2015fast}
Bosdriesz E, Molenaar D, Teusink B, Bruggeman FJ.
\newblock How fast-growing bacteria robustly tune their ribosome concentration
  to approximate growth-rate maximization.
\newblock The FEBS journal. 2015;282(10):2029--2044.

\bibitem{scott2010interdependence}
Scott M, Gunderson CW, Mateescu EM, Zhang Z, Hwa T.
\newblock Interdependence of cell growth and gene expression: origins and
  consequences.
\newblock Science. 2010;330(6007):1099--1102.

\bibitem{scott_emergence_2014}
Scott M, Klumpp S, Mateescu EM, Hwa T.
\newblock Emergence of robust growth laws from optimal regulation of ribosome
  synthesis.
\newblock Molecular Systems Biology. 2014;10(8):747.
\newblock doi:{10.15252/msb.20145379}.

\bibitem{dai2020coupling}
Dai X, Zhu M.
\newblock Coupling of ribosome synthesis and translational capacity with cell
  growth.
\newblock Trends in biochemical sciences. 2020;45(8):681--692.

\bibitem{roy2021unifying}
Roy A, Goberman D, Pugatch R.
\newblock A unifying autocatalytic network-based framework for bacterial growth
  laws.
\newblock Proceedings of the National Academy of Sciences.
  2021;118(33):e2107829118.

\bibitem{Calabrese2023}
Calabrese L, Ciandrini L, Lagomarsino MC.
\newblock How total mRNA influences cell growth.
\newblock bioRxiv. 2023;doi:{10.1101/2023.03.17.533181}.

\bibitem{dai_reduction_2017}
Dai X, Zhu M, Warren M, Balakrishnan R, Patsalo V, Okano H, et~al.
\newblock Reduction of translating ribosomes enables Escherichia coli to
  maintain elongation rates during slow growth.
\newblock Nature Microbiology. 2017;2(2):16231.
\newblock doi:{10.1038/nmicrobiol.2016.231}.

\bibitem{wu2022cellular}
Wu C, Balakrishnan R, Braniff N, Mori M, Manzanarez G, Zhang Z, et~al.
\newblock Cellular perception of growth rate and the mechanistic origin of
  bacterial growth law.
\newblock Proceedings of the National Academy of Sciences.
  2022;119(20):e2201585119.

\bibitem{metzl2017principles}
Metzl-Raz E, Kafri M, Yaakov G, Soifer I, Gurvich Y, Barkai N.
\newblock Principles of cellular resource allocation revealed by
  condition-dependent proteome profiling.
\newblock Elife. 2017;6:e28034.

\bibitem{xia2022proteome}
Xia J, S{\'a}nchez BJ, Chen Y, Campbell K, Kasvandik S, Nielsen J.
\newblock Proteome allocations change linearly with the specific growth rate of
  Saccharomyces cerevisiae under glucose limitation.
\newblock Nature Communications. 2022;13(1):2819.

\bibitem{calabrese2022protein}
Calabrese L, Grilli J, Osella M, Kempes CP, Lagomarsino MC, Ciandrini L.
\newblock Protein degradation sets the fraction of active ribosomes at
  vanishing growth.
\newblock PLoS computational biology. 2022;18(5):e1010059.

\bibitem{gupta2022global}
Gupta M, Johnson A, Cruz E, Costa E, Guest R, Li SHJ, et~al.
\newblock Global Protein-Turnover Quantification in Escherichia coli Reveals
  Cytoplasmic Recycling under Nitrogen Limitation.
\newblock bioRxiv. 2022; p. 2022--08.

\bibitem{wyant2018nufip1}
Wyant GA, Abu-Remaileh M, Frenkel EM, Laqtom NN, Dharamdasani V, Lewis CA,
  et~al.
\newblock NUFIP1 is a ribosome receptor for starvation-induced ribophagy.
\newblock Science. 2018;360(6390):751--758.

\bibitem{nakatogawa2018spoon}
Nakatogawa H.
\newblock Spoon-feeding ribosomes to autophagy.
\newblock Molecular cell. 2018;71(2):197--199.

\bibitem{sinturel2017diurnal}
Sinturel F, Gerber A, Mauvoisin D, Wang J, Gatfield D, Stubblefield JJ, et~al.
\newblock Diurnal oscillations in liver mass and cell size accompany ribosome
  assembly cycles.
\newblock Cell. 2017;169(4):651--663.

\bibitem{bremer1975parameters}
Bremer H.
\newblock Parameters affecting the rate of synthesis of ribosomes and RNA
  polymerase in bacteria.
\newblock Journal of Theoretical Biology. 1975;53(1):115--124.

\bibitem{chure2023optimal}
Chure G, Cremer J.
\newblock An optimal regulation of fluxes dictates microbial growth in and out
  of steady state.
\newblock Elife. 2023;12:e84878.

\bibitem{kostinski2020ribosome}
Kostinski S, Reuveni S.
\newblock Ribosome composition maximizes cellular growth rates in E. coli.
\newblock Physical review letters. 2020;125(2):028103.

\bibitem{kostinski2021growth}
Kostinski S, Reuveni S.
\newblock Growth laws and invariants from ribosome biogenesis in lower eukarya.
\newblock Physical Review Research. 2021;3(1):013020.

\bibitem{basan2015overflow}
Basan M, Hui S, Okano H, Zhang Z, Shen Y, Williamson JR, et~al.
\newblock Overflow metabolism in Escherichia coli results from efficient
  proteome allocation.
\newblock Nature. 2015;528(7580):99--104.

\bibitem{klumpp_molecular_2013}
Klumpp S, Scott M, Pedersen S, Hwa T.
\newblock Molecular crowding limits translation and cell growth.
\newblock Proceedings of the National Academy of Sciences.
  2013;110(42):16754--16759.
\newblock doi:{10.1073/pnas.1310377110}.

\bibitem{fraser2004noise}
Fraser HB, Hirsh AE, Giaever G, Kumm J, Eisen MB.
\newblock Noise Minimization in Eukaryotic Gene Expression.
\newblock PLOS Biology. 2004;2(6):null.
\newblock doi:{10.1371/journal.pbio.0020137}.

\bibitem{joshi_what_2022}
Joshi CJ, Ke W, Drangowska-Way A, O’Rourke EJ, Lewis NE.
\newblock What are housekeeping genes?
\newblock {PLOS} Computational Biology. 2022;18(7):e1010295.
\newblock doi:{10.1371/journal.pcbi.1010295}.

\bibitem{guckenheimer_nonlinear_1983}
Guckenheimer J, Holmes P.
\newblock Nonlinear Oscillations, Dynamical Systems, and Bifurcations of Vector
  Fields. vol.~42 of Applied Mathematical Sciences.
\newblock Springer New York; 1983.
\newblock Available from:
  \url{http://link.springer.com/10.1007/978-1-4612-1140-2}.

\bibitem{2020SciPy-NMeth}
Virtanen P, Gommers R, Oliphant TE, Haberland M, Reddy T, Cournapeau D, et~al.
\newblock {{SciPy} 1.0: Fundamental Algorithms for Scientific Computing in
  Python}.
\newblock Nature Methods. 2020;17:261--272.
\newblock doi:{10.1038/s41592-019-0686-2}.

\bibitem{wachsmuth2012optimal}
Wachsmuth G.
\newblock Optimal control of quasi-static plasticity with linear kinematic
  hardening, Part I: Existence and discretization in time.
\newblock SIAM Journal on Control and Optimization. 2012;50(5):2836--2861.

\bibitem{swan1981optimal}
Swan GW.
\newblock Optimal control applications in biomedical engineering—a survey.
\newblock Optimal Control Applications and Methods. 1981;2(4):311--334.

\bibitem{ewald2020trends}
Ewald J, Sieber P, Garde R, Lang SN, Schuster S, Ibrahim B.
\newblock Trends in mathematical modeling of host--pathogen interactions.
\newblock Cellular and Molecular Life Sciences. 2020;77:467--480.

\bibitem{lunz2023optimal}
Lunz D, Bonnans JF, Ruess J.
\newblock Optimal control of bioproduction in the presence of population
  heterogeneity.
\newblock Journal of Mathematical Biology. 2023;86(3):43.

\bibitem{giordano2016dynamical}
Giordano N, Mairet F, Gouz{\'e} JL, Geiselmann J, De~Jong H.
\newblock Dynamical allocation of cellular resources as an optimal control
  problem: novel insights into microbial growth strategies.
\newblock PLoS computational biology. 2016;12(3):e1004802.

\bibitem{klipp2002prediction}
Klipp E, Heinrich R, Holzh{\"u}tter HG.
\newblock Prediction of temporal gene expression: Metabolic optimization by
  re-distribution of enzyme activities.
\newblock European journal of biochemistry. 2002;269(22):5406--5413.

\bibitem{dekel2005optimality}
Dekel E, Alon U.
\newblock Optimality and evolutionary tuning of the expression level of a
  protein.
\newblock Nature. 2005;436(7050):588--592.

\bibitem{poelwijk2011optimality}
Poelwijk FJ, Heyning PD, de~Vos MG, Kiviet DJ, Tans SJ.
\newblock Optimality and evolution of transcriptionally regulated gene
  expression.
\newblock BMC systems biology. 2011;5(1):1--12.

\bibitem{pavlov2013optimal}
Pavlov MY, Ehrenberg M.
\newblock Optimal control of gene expression for fast proteome adaptation to
  environmental change.
\newblock Proceedings of the National Academy of Sciences.
  2013;110(51):20527--20532.

\bibitem{tsiantis2020using}
Tsiantis N, Banga JR.
\newblock Using optimal control to understand complex metabolic pathways.
\newblock BMC bioinformatics. 2020;21(1):1--33.

\bibitem{tsokanos2016eif}
Tsokanos FF, Albert MA, Demetriades C, Spirohn K, Boutros M, Teleman AA.
\newblock eIF 4A inactivates TORC 1 in response to amino acid starvation.
\newblock The EMBO journal. 2016;35(10):1058--1076.

\bibitem{kawamata_zinc_2017}
Kawamata T, Horie T, Matsunami M, Sasaki M, Ohsumi Y.
\newblock Zinc starvation induces autophagy in yeast.
\newblock Journal of Biological Chemistry. 2017;292(20):8520--8530.
\newblock doi:{10.1074/jbc.M116.762948}.

\bibitem{sharp2021implementation}
Sharp JA, Burrage K, Simpson MJ.
\newblock Implementation and acceleration of optimal control for systems
  biology.
\newblock Journal of the Royal Society Interface. 2021;18(181):20210241.

\bibitem{gao2021quasi}
Gao L, Dai X, Kleeberger M, Fottner J.
\newblock Quasi-static Optimal Control Strategy of Lattice Boom Crane Based on
  Large-Scale Flexible Non-linear Dynamics.
\newblock In: International Conference on Simulation and Modeling
  Methodologies, Technologies and Applications. Springer; 2021. p. 153--177.

\bibitem{sabi2019modelling}
Sabi R, Tuller T.
\newblock Modelling and measuring intracellular competition for finite
  resources during gene expression.
\newblock Journal of the Royal Society Interface. 2019;16(154):20180887.

\bibitem{katz2022translation}
Katz R, Attias E, Tuller T, Margaliot M.
\newblock Translation in the cell under fierce competition for shared
  resources: a mathematical model.
\newblock Journal of the Royal Society Interface. 2022;19(197):20220535.

\bibitem{zarai2017controllability}
Zarai Y, Margaliot M, Sontag ED, Tuller T.
\newblock Controllability analysis and control synthesis for the ribosome flow
  model.
\newblock IEEE/ACM Transactions on Computational Biology and Bioinformatics.
  2017;15(4):1351--1364.

\bibitem{scott2023shaping}
Scott M, Hwa T.
\newblock Shaping bacterial gene expression by physiological and proteome
  allocation constraints.
\newblock Nature Reviews Microbiology. 2023;21(5):327--342.

\bibitem{rosario2020mtorc1}
Rosario FJ, Powell TL, Gupta MB, Cox L, Jansson T.
\newblock mTORC1 transcriptional regulation of ribosome subunits, protein
  synthesis, and molecular transport in primary human trophoblast cells.
\newblock Frontiers in cell and developmental biology. 2020; p. 1301.

\bibitem{demetriades2014regulation}
Demetriades C, Doumpas N, Teleman AA.
\newblock Regulation of TORC1 in response to amino acid starvation via
  lysosomal recruitment of TSC2.
\newblock Cell. 2014;156(4):786--799.

\bibitem{dao2018impact}
Dao~Duc K, Song YS.
\newblock The impact of ribosomal interference, codon usage, and exit tunnel
  interactions on translation elongation rate variation.
\newblock PLoS genetics. 2018;14(1):e1007166.

\bibitem{Meiyan2018Ribophagy}
Jin M, Klionsky DJ.
\newblock Finding a ribophagy receptor.
\newblock Autophagy. 2018;14(9):1479--1480.
\newblock doi:{10.1080/15548627.2018.1483672}.

\bibitem{an_ribosome_2020}
An H, Harper JW.
\newblock Ribosome Abundance Control Via the Ubiquitin–Proteasome System and
  Autophagy.
\newblock Journal of Molecular Biology. 2020;432(1):170--184.
\newblock doi:{10.1016/j.jmb.2019.06.001}.

\bibitem{khasminskii2011stochastic}
Khasminskii R.
\newblock Stochastic stability of differential equations. vol.~66.
\newblock Springer Science \& Business Media; 2011.

\bibitem{buskirk2017ribosome}
Buskirk AR, Green R.
\newblock Ribosome pausing, arrest and rescue in bacteria and eukaryotes.
\newblock Philosophical Transactions of the Royal Society B: Biological
  Sciences. 2017;372(1716):20160183.

\bibitem{chen1999modeling}
Chen T, He HL, Church GM.
\newblock Modeling gene expression with differential equations.
\newblock In: Biocomputing'99. World Scientific; 1999. p. 29--40.

\bibitem{witelski2015methods}
Witelski T, Bowen M.
\newblock Methods of mathematical modelling.
\newblock Springer; 2015.

\bibitem{roussel1996use}
Roussel MR.
\newblock The use of delay differential equations in chemical kinetics.
\newblock The journal of physical chemistry. 1996;100(20):8323--8330.

\bibitem{kressler_driving_2010}
Kressler D, Hurt E, Ba{\ss}ler J.
\newblock Driving ribosome assembly.
\newblock Biochimica et Biophysica Acta ({BBA}) - Molecular Cell Research.
  2010;1803(6):673--683.
\newblock doi:{10.1016/j.bbamcr.2009.10.009}.

\bibitem{nelson2013dynamical}
Nelson P.
\newblock Dynamical Systems Theory, Delay Differential Equations.
\newblock Encyclopedia of Systems Biology. 2013; p. 637--641.

\bibitem{teschl2012ordinary}
Teschl G.
\newblock Ordinary differential equations and dynamical systems. vol. 140.
\newblock American Mathematical Soc.; 2012.

\bibitem{dumortier_qualitative_2006}
Dumortier F, Llibre J, Artés JC.
\newblock Qualitative Theory of Planar Differential Systems.
\newblock Applied Mathematical Sciences. Springer-Verlag Berlin Heidelberg;
  2006.
\newblock Available from:
  \url{http://link.springer.com/10.1007/978-1-4612-1140-2}.

\bibitem{hirsch_turnover_1966}
Hirsch CA, Hiatt HH.
\newblock Turnover of Liver Ribosomes in Fed and in Fasted Rats.
\newblock Journal of Biological Chemistry. 1996;241(24):5936--5940.
\newblock doi:{10.1016/S0021-9258(18)96360-X}.

\bibitem{encyclopedia-2004}
Lennarz WJ, Lane MD.
\newblock Encyclopedia of Biological Chemistry.
\newblock Elsevier; 2004.

\bibitem{zhang_generic_2009}
Zhang G, Ignatova Z.
\newblock Generic algorithm to predict the speed of translational elongation:
  implications for protein biogenesis.
\newblock {PloS} one. 2009;4(4):e5036.
\newblock doi:{10.1371/journal.pone.0005036}.

\bibitem{Albe_2002_book}
Alberts B, Bray D, Lewis J, Raff M, Roberts K, Watson JD.
\newblock Molecular Biology of the Cell.
\newblock 4th ed. Garland; 2002.

\bibitem{belliveau_fundamental_2020}
Belliveau NM, Chure G, Hueschen CL, Garcia HG, Kondev J, Fisher DS, et~al.
\newblock Fundamental limits on the rate of bacterial growth.
\newblock bioRxiv. 2020;doi:{10.1101/2020.10.18.344382}.

\end{thebibliography}

%%%%%%%%%%%%%%%%%%%%%%%%%%%%%%%%%%%%%%%%%%%%%%%%%%%%%%%%%%%%%%%%%%%%%%%%%%
\newpage
%%%%%%%%%%%%%%%%%%%%%%%%%%%%%%%%%%%%%%%%%%%%%%%%%%%%%%%%%%%%%%%%%%%%%%%%%%
\begin{appendices}
%%%%%%%%%%%%%%%%%%%%%%%%%%%%%%%%%%%%%%%%%%%%%%%%%%%%%%%%%%%%%%%%%%%%%%%%%%
\section{Global asymptotic stability of the system }\label{sec:appendix_stability}
%%%%%%%%%%%%%%%%%%%%%%%%%%%%%%%%%%%%%%%%%%%%%%%%%%%%%%%%%%%%%%%%%%%%%%%%%%

We prove in this Appendix section proposition \ref{prop:stability} from the main text (stated here again, for clarity):

\begin{proposition*} We assume $\alpha$ constant and the rate functions $\gamma_{i}(E)$ ($i=P,R$) satisfying (i) $\gamma_i(0)=0$, (ii) $\gamma_{i}$ is $\mathcal{C}^1$ with positive derivative $\gamma_i'>0$, and (iii) $\lim_{E\rightarrow \infty} \gamma_i(E) = \gamma_{i,\max} < \infty$.
Then, the solutions of the autonomous system \eqref{eq:dRdt}-\eqref{eq:dEdt} starting at $E(0), \ R(0) >0$ remain strictly positive for all time, and converge to a globally asymptotically stable equilibrium point 
$(R^*,E^*)$  on $\R_{>0}^2$, given by
\beq
R^* &=& \frac{\alpha}{c_R(\beta_r+U)+c_P\gamma_P(E^*)(1-V)-c_U U}\label{eq:R*}\\
E^* &=& \gamma_R^{-1}\left(\frac{\beta_r+U}{V}\right) \label{eq:E*}.
\eeq
Moreover, all solutions initialized in $\R_{>0}^2$ oscillate, in the sense that $R(t)-R^*$ changes sign infinitely often as $t\to\infty$, if and only if $\Delta<0$, where
\beq
\Delta=(c_R\gamma_R'(E^*)V+c_P\gamma_P'(E^*)(1-V))^2R^{*2}-4\alpha \gamma_R'(E^*)V.
\eeq
\end{proposition*}

\begin{proof} We first get $(R^*,E^*)$ by solving for the equilibrium point of the system with constant input and control functions. 
To assess the stability of $(R^*,E^*)$, we study the Jacobian of the system at $(R^*,E^*)$:

$$J=\pmtx 0 & \gamma_R'(E^*)V R^*\\
-\frac{\alpha}{R^*} & -(c_R\gamma_R'(E^*)V+c_P\gamma_P'(E^*)(1-V))R^*
\pmtrx.$$
The characteristic polynomial of $J$ is 
\begin{equation}\label{eq:charac_poly}
     C[X]= X^2 +(c_R\gamma_R'(E^*)V+c_P\gamma_P'(E^*)(1-V))R^* X+\alpha \gamma_R'(E^*)V,
\end{equation}
with discriminant 
$$\Delta=(c_R\gamma_R'(E^*)V+c_P\gamma_P'(E^*)(1-V))^2R^{*2}-4 \alpha \gamma_R'(E^*)V.$$
If $\Delta<0,$ the eigenvalues of $J$ have real part $-\frac{1}{2}(c_R\gamma_R'(E^*)V+c_P\gamma_P'(E^*)(1-V))R^*<0$.
If $\Delta\geq0$, eigenvalues are real and strictly negative, since the coefficients of $C[X]$ are positive.
In both cases, we can apply the Hartman-Grobman theorem \cite{guckenheimer_nonlinear_1983} to conclude that $(R^*,E^*)$ is locally asymptotically stable.\newline

%%%%%%%%%%%%%%%%%%%%%%%%%%%%%%%%%%%%%%%%%%%%%%%%%%%%%%%%%%%%%%%%%%%%%%%%%%
\begin{figure}[!ht]
    \centering
         \includegraphics[width=\textwidth]{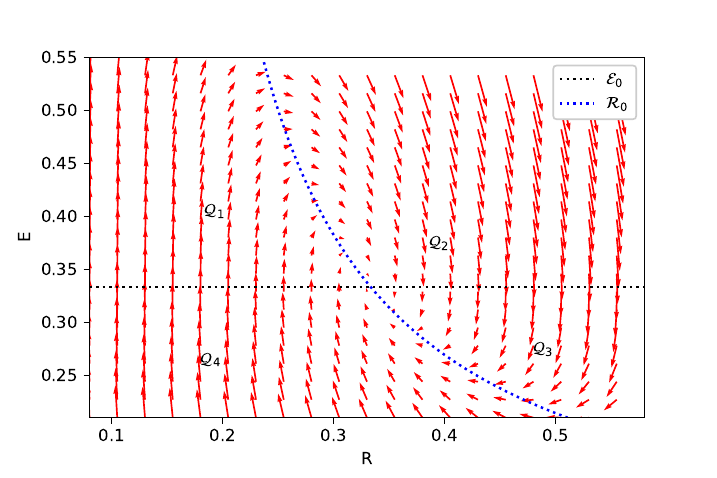}
        \caption{Vector field associated with the dynamical system. The $R$- and $E$-nullclines (equations \eqref{eq:Rnullcline}-\eqref{eq:Enullcline}) represented in blue and black dotted lines, respectively. They partition $\mathbb{R}_{>0}^2$ into four regions $\cQ_i$, $1 \le i \le 4$ (Eq. \eqref{eq:Qi}), and we also represent the section $S$, introduced in Eq. \eqref{eq:S} to define the Poincaré map and prove the asymptotic convergence of trajectories starting from $S$.}
        \label{fig:Q1234}
\end{figure}
%%%%%%%%%%%%%%%%%%%%%%%%%%%%%%%%%%%%%%%%%%%%%%%%%%%%%%%%%%%%%%%%%%%%%%%%%%
To prove that $(R^*,E^*)$ is globally asymptotically stable on $\mathbb{R}_{>0}^2$, we examine the nullclines to partition the state space. The $R$- and $E$-nullclines of the system, on which $R'=0$, respectively, $E'=0$, are given by 
\begin{align}
\label{eq:Rnullcline}
  \mathcal{R}_0 &= \{(R, E^*), \ R>0 \}\\
  \mathcal{E}_0 &= \{(R_{\mathcal{E}_0}(E),E), \ E>E_{\cE_0} \}\label{eq:Enullcline},
\end{align}
where 
\begin{align}\label{eq:RE0E}
R_{\mathcal{E}_0}(E)=\frac{\alpha}{c_R \gamma_R(E)V + c_P \gamma_P (E)(1-V) - c_U U}
\end{align}
and $E_{\cE_0}$ is either the unique positive value of $E$ making the above denominator $0$, or $0$ if no such value exists. Note that $E_{\cE_0}<E^*$, and that $R_{\cE_0}$ is positive on its domain and decreases with $E$, since $\gamma_R,\gamma_P$ increase with $E$. As illustrated in Figure \ref{fig:Q1234}, we can thus partition $\R_{>0}^2$ into the disjoint union of the four open sets $\cQ_i$, $1 \le i \le 4$ given by
\begin{equation}
\label{eq:Qi}
\cQ_i := \{(R,E)\in \R_{>0}^2\colon \sgn(f_1(R,E),f_2(R,E)) = a_i\},
\end{equation}
where $a_1 = (1,1), \ a_2=(1,-1), \ a_3 = (-1,-1), \ a_4=(-1,1),$ and $f_1$ ($f_2$) is the function from equation \eqref{eq:dRdt} (\eqref{eq:dEdt}) associated with $dR/dt$ ($dE/dt$). We further introduce \beq\label{eq:S}\mathcal{S}=\left\{ (R,E^*) \ , \  0<R< R^*\right\}\eeq as the part of the $R$-nullcline that separates $\mathcal{Q}_1$ and $\mathcal{Q}_4$. Solutions of the system then verify the following lemma

\begin{lemma}\label{trajectories}

Let $\mathbf{X}(t)$ be a solution of the system with initial conditions $\mathbf{X}(0) \in \mathbb{R}_{+*}^2$ and let $t_s=\inf\{t\ge 0\colon \mathbf{X}(t) \in S\}$. If $t_s=\infty$ then $\lim_{t\to\infty} X(t) = (R^*,E^*)$.
\end{lemma}

\begin{proof}
Since the vector field $(f_1,f_2)$ is transverse to $\cR_0\cup \cE_0$ away from $(R^*,E^*)$, and looking at the direction of the vector field in Figure \ref{fig:Q1234}, if $\mathbf{X}(t) \in \cR_0\cup \cE_0 \backslash \{(R^*,E^*)\}$, then for some $t'<t<t''$ and some $i$, $\mathbf{X}(s) \in \cQ_i$ for $t'<s<t$ and $\mathbf{X}(s) \in \cQ_{i+1}$ (or $\cQ_1$ if $i=4$) for $t<s<t''$. So, defining $t_0=0$ and recursively, $t_{j+1}=\inf\{t>t_j\colon \mathbf{X}(t) \in \cR_0\cup \cE_0\}$, if $t_j<\infty$ and $j>0$ then $\mathbf{X}(t_j) \in \cR_0 \cup \cE_0$, and if $t_j<\infty$ then $\mathbf{X}(t) \in \cQ_{i(j)}$ for $t\in(t_j,t_{j+1})$, where $i(j)=i(0)+j \mod 4$. In other words, trajectories follow the itinerary $\cQ_1 \to \cQ_2\to \cQ_3\to \cQ_4\to \cQ_1$, although they may, after arriving in some $\cQ_i$, remain there forever. To prove the lemma, and  since $t_s$ is the first time the trajectory crosses from $\cQ_4$ into $\cQ_1$, we can then  show that for all $i=1,\ldots,4 $, if  $\mathbf{X}(0) \in \cQ_i$ and $t_1=\infty$ (i.e., if $\mathbf{X}(t)\in \cQ_i$ for all $t>0$) then $\lim_{t\to\infty}\mathbf{X}(t)=(R^*,E^*)$: 
\begin{enumerate}
    \item 
    $i=1$: $t\mapsto E(t)$ and $t\mapsto R(t)$ are increasing for $t\in [0,t_1)$. Since $f_1(R,E)=(\gamma_R(E)V-(\beta_R+U))R$, $f_1$ is increasing in both $R,E$ on the set $\{f_1>0\}$, so $f_1(R(t),E(t))\ge f_1(R(0),E(0))=:c>0$ and $R(t) \ge R(0)+ct$ for $0<t<t_1$. Since $R\le R^*$ on $\cQ_1$, it follows that $t_1 \le (R^*-R(0))/c$. In particular, $t_1<\infty$.
    \item 
    $i=2$: $t\mapsto R(t)$ is increasing and $t\mapsto E(t)$ is decreasing for $t\in [0,t_1)$. In particular, $E(0)\ge E(t)\ge E^*$ for all $t<t_1$, so if $t_1=\infty$ then $E_\infty:=\lim_{t\to\infty}E(t)$ exists. If $t_1=\infty$ and $R(t)\le R^*$ for all $t>0$ then $R_\infty:=\lim_{t\to\infty} R(t)$ exists. By continuity of the vector field, $(R_\infty,E_\infty)$ is an equilibrium point, so must be equal to $(R^*,E^*)$. If $R(t)>R^*$ for some $t\in [0,t_1)$, then for $s>t$,
    $$E'(s)\le \alpha + \left[- c_R \gamma_R(E^*)V - c_P \gamma_P (E^*)(1-V) + c_U U\right]R(t)=:c<0.$$
    Arguing as in Case 1, $t_1\le t + (E(t)-E^*)/c<\infty$.
    \item  $i=3$: $t\mapsto R(t)$, $t\mapsto E(t)$ are both decreasing for $t\in [0,t_1)$, and $E(0)<E^*$. Since $\R_{>0}^2$ is forward invariant, $R(t)\ge 0$ and $E(t)\ge 0$ for all $t>0$ so if $t_1=\infty$ then $(R_\infty,E_\infty)$ exist as before, so must be equal to $(R^*,E^*)$. But $E_\infty\le E(0)<E^*$ which is a contradiction, so $t_1<\infty$.
    \item  $i=4$: $t\mapsto R(t)$ is decreasing and $t\mapsto E(t)$ is increasing for $t\in [0,t_1)$. If $t\in [0,t_1)$ then $E(t) \le E^*$, and $R(t)\ge 0$ for all $t>0$, so if $t_1=\infty$ then $(R_\infty,E_\infty)$ exist and must be equal to $(R^*,E^*)$.
\end{enumerate}
\end{proof}
To prove $(E^*,R^*)$'s global stability, we now only need to show that trajectories starting at and crossing $\mathcal{S}$  eventually converge to $(R^*,E^*)$. We parameterize these trajectories by their initial position $r\in \mathcal{I}$, where  

\begin{equation}
 \mathcal{I} = \left\{ r \in ]0; R^*[ \ |  \ \mathbf{X}(0) = \pmtx r \\ E^*\pmtrx \Rightarrow \exists t>0, \ \mathbf{X}(t)\in \mathcal{S}  \right\}.  
\end{equation}
Assuming that  $\mathcal{I}$ is non-empty (if not, we can directly conclude using Lemma \ref{trajectories} that $(R^*,E^*)$ is globally asymptotically stable), we define on $\mathcal{I} \times \R_*$ the application
\begin{equation}
\begin{array}{cccccccc}
    \Phi: & \mathcal{I} \times \R_+ & \to & \R_{>0}^2 &\\
    &(r,t) & \mapsto  &  \pmtx \phi_1(r,t) \\ \phi_2(r,t)\pmtrx & =\mathbf{X}_r(t),
\end{array}
\end{equation}
where $\mathbf{X}_r$ is a solution of the system with initial condition $\mathbf{X}_r(0)= \pmtx r \\ E^*\pmtrx$. This allows us to define the so-called Poincaré map \cite{teschl2012ordinary}  associated with the system \eqref{eq:dEdt}-\eqref{eq:dRdt} and section $\mathcal{S}$, as
\begin{equation}
    \begin{array}{ccccccc}
    \mathcal{P}: & \mathcal{I} & \to & ]0;R^*[  \\
    & r  & \mapsto  & \phi_1(r,t_r)\label{eq:poincaremap}
\end{array}
\end{equation}
where $t_r = \inf\{t>0, \ \mathbf{X}_r(t) \in \mathcal{S}\}$ is the first return time to $S$ ($t_r$ exists by definition of $\mathcal{I}$).

\begin{lemma}\label{prop:Poincareprop}
The Poincaré map $\mathcal{P}$, as defined in equation \eqref{eq:poincaremap}, satisfies the following properties:
\begin{enumerate}[label=(\roman*)]
    \item $\mathcal{P}$ is $\mathcal{C}^1$ and $\exists \ 0< r_\mathcal{P}\leq R^*, \ \mathcal{I}= ]0,r_\mathcal{P}[$ or $\mathcal{I}= ]0,r_\mathcal{P}]$.
    \item $\mathcal{P}$ has no fixed point.
\end{enumerate}
\end{lemma}

\begin{proof}
\begin{enumerate}[label=(\roman*)]
    \item \label{prop:Poincareprop-i} Let $r_0 \in \mathcal{I}$ and $t_r$ its first return time, so $\phi_2(r_0,t_r)=E^*$ and $\phi_1(r_0,t_r)\in (0,R^*)$. Since the vector field is transverse to the nullclines away from $(R^*,E^*)$, $\partial_t\phi_2(r,t_r)\ne 0$. Since $\phi_2$ is $\mathcal{C}^1$, we can apply the implicit function theorem: There exists an open set $V$ containing $r_0$, an open set $\Omega$ containing $(r_0,t_r)$ and a $\mathrm{C}^1$ function $\xi$ such as : $$((r,t)\in \Omega \And \phi_2(r,t)=E^*) \iff (r\in V \And t=\xi(r) ).$$ 
    To conclude that $\mathcal{P}(r)=\phi_1(r,\xi(r))$, we would need to know that $\xi(r)$ is the first return time of $(r,E^*)$ to $S$. This can be shown, however, for $r$ in some open $U\subset V$ containing $r_0$, using continuity of the flow together with local injectivity of $(r,t)\mapsto \phi(r,t)$ near $(r_0,0)$. In particular, the domain of $\mathcal{P}$ is open, so $\mathcal{I}$ is an interval. Since $\phi_1$ and $\xi$ are $\mathcal{C}^1$ and $r_0 \in \mathcal{I}$ is arbitrary, it follows that $\mathcal{P}$ is $\mathcal{C}^1$ on $\mathcal{I}$.
    
    \item \label{prop:Poincareprop-ii}This is a consequence of the Bendixson-Dulac theorem \cite{dumortier_qualitative_2006}. Let's define $\nu(R,E)=\frac{1}{R}$, $f(R,E)=\frac{d R}{dt}$ and $g(R,E)=\frac{d E}{dt}$. Then,
        $$\frac{\partial(\nu f)}{\partial R}+\frac{\partial(\nu g)}{\partial E}=-(c_P \gamma_P'(E)(1-V) +c_R \gamma_R'(E)V) $$
    which is strictly negative (as we assume $\gamma_P$ and $\gamma_R$ being strictly increasing). Upon applying the Bendixson-Dulac theorem, we find the system has no nonconstant periodic solutions, and in particular, $\mathcal{P}$ admits no fixed point.
\end{enumerate}
\end{proof}
We can now characterize the behaviour of the iterated map, as

\begin{lemma}
Solutions $\mathbf{X}_r(t)$ of the system that start from $(r,E^*)$, where $r\in \mathcal{I}$, converge to  $(R^*,E^*)$.
\end{lemma}

\begin{proof} Using lemma \ref{prop:Poincareprop}, either $\forall r \in \mathcal{I}, \ \mathcal{P}(r)>r$, or $\forall r \in \mathcal{I}, \ \mathcal{P}(r)<r$. Let us consider again $r_\mathcal{P}$, as defined in lemma \ref{prop:Poincareprop}.

If $r_\mathcal{P} =  R^*$, then there exists $r_1$ such as $\mathcal{P}(r_1)>r_1$, since $(R^*,E^*)$ is locally asymptotically stable. Thus, for all $r\in\mathcal{I}$, the iterated sequence $\mathcal{P}^n(r)$ is strictly increasing, and it converges to $R^*$ (another limit of the sequence would be a fixed point of $\mathcal{P}$, which contradicts proposition \ref{prop:Poincareprop}), and $\mathbf{X}_r(t)$ converges to $(R^*,E^*)$. \\
If $r_\mathcal{P} <  R^*$, let $r_\mathcal{P}<r_0<R^*$. By taking the reversal trajectory $\Phi(r_0,-t)$, and using similar arguments as in lemma \ref{trajectories}, we can prove that there exists a first time $t>0$, such that $\phi_1(r_0,-t)\in \mathcal{I} $. By definition, $\mathcal{P}(\phi_1(r_0,-t))=r_0>\phi_1(r_0,-t)$, so $\mathcal{P}(r)>r$ for all $r\in \mathcal{I}$. In particular, there exists $n$ such that $\mathcal{P}^n(r) > r_\mathcal{P}$, i.e. $\mathcal{P}^n(r) \notin \mathcal{I}$. Using lemma \ref{trajectories}, $\mathbf{X}_r(t)$ thus converges to $(R^*,E^*)$ 
.\end{proof}

Combined with lemma \ref{trajectories}, this last lemma shows that $(R^*,E^*)$ is globally asymptotically stable. Finally, the second part of the proposition is a direct consequence of the stable manifold theorem for a fixed point \cite{guckenheimer_nonlinear_1983}, showing that if $\Delta\geq0$ the eigenvalues are real and there is a finite number of oscillations produced as the trajectories converge, and if $\Delta<0$, eigenvalues have a non zero imaginary part leading to an infinite number of oscillations around the equilibrium point.

\end{proof}

%%%%%%%%%%%%%%%%%%%%%%%%%%%%%%%%%%%%%%%%%%%%%%%%%%%%%%%%%%%%%%%%%%%%%%%%%%
\newpage
%%%%%%%%%%%%%%%%%%%%%%%%%%%%%%%%%%%%%%%%%%%%%%%%%%%%%%%%%%%%%%%%%%%%%%%%%%
\section{Convergence behavior}\label{appendix_2}
%%%%%%%%%%%%%%%%%%%%%%%%%%%%%%%%%%%%%%%%%%%%%%%%%%%%%%%%%%%%%%%%%%%%%%%%%%
In this Appendix section, we study the behaviour of the autonomous system with constant $\alpha$, $U$ and $V$ such that conditions \eqref{eq:mass_conservation} and \eqref{eq:constraint2} hold.

\begin{proposition}\label{prop:oscillation}
    There exist two oscillation threshold functions $\alpha_{\osc}(V)$ and $U_{\osc}(\alpha,V)$ such that:
    \begin{enumerate}[label=(\roman*)]
        \item if $\alpha<\alpha_{\osc}(V)$, the system is infinitely oscillating around the steady state,
        \item if $\alpha\geq\alpha_{\osc}(V)$, the system is infinitely oscillating around the steady state if and only if $U>U_{osc}$.
    \end{enumerate}

\end{proposition}

\begin{proof} 
Solving $\Delta=0$ for $U$ with the values of $(R^*,E^*)$ from equations \eqref{eq:R*}-\eqref{eq:E*} leads to 
\begin{equation}\label{def_uosc}
       U= U_{\osc}:=\frac{C_0\alpha (\gamma_{R,\max}V-\beta_R)-k_D\beta_R( c_P\gamma_{P,\max}(1-V)+c_R\gamma_{R,\max}V)}{k_D( c_P\gamma_{P,\max}(1-V)+c_R\gamma_{R,\max}V)+C_0\alpha-k_D\gamma_{R,\max}Vc_U},
 \end{equation}
 where  $$C_0= \frac{k_D+E^*}{R^*}=\frac{c_P\gamma_{P,\max}k_D(1-V)+c_R \gamma_{R,\max}k_D V}{(4\alpha\gamma_{R,\max} k_D V)^{\frac{1}{2}}}.$$ For this solution to be physical, we need $U_{\osc} \geq 0$. Note that the denominator is always positive because of condition \ref{eq:mass_conservation}, so $U_{\osc} \geq 0$ yields
 
\begin{equation}\label{alpha_osc_part}
    \alpha > \alpha_{\osc}\coloneqq\frac{4\beta_R^2k_D\gamma_{R,\max}V}{(\gamma_{R,\max}V-\beta_R)^2}.
\end{equation}
In particular, if $\alpha \leq \alpha_{osc}$, $\Delta <0 $ and $(i)$ is proved. Assuming $\alpha > \alpha_{osc}$, solving $\Delta(U)<0$ similarly leads to $U> U_{\osc}$.
\end{proof}

\begin{proposition}\label{prop:Psi_variation}

Let $\alpha >0$ and $\Psi:=-\real(\lambda)$, where $\lambda$ is the eigenvalue with the largest real part of the Jacobian matrix of the system at $(R^*,E^*)$. 

\begin{enumerate}[label=(\roman*)]
\item Let $U_{osc}^*=U_{osc}(V=1)$ (as defined in Eq. \eqref{def_uosc}). If $U_{osc}^* >0$ and $U\leq  U_{osc}^*$, then $\Psi(V)$ is increasing  as $U_\osc(V)\leq U  $, and is decreasing for $U_\osc(V)>U $. Else, it is only increasing. 

\item
Let $V\in[0,1]$, then, if $\alpha \geq \alpha_\osc(V)$, $\Psi(U)$ is increasing as $U \leq U_\osc(V)$, and is decreasing for $U > U_\osc(V)$. Else, it is only decreasing.
\end{enumerate}

\end{proposition}

\begin{proof}   We first show that  $U_{\osc}(V)$ is increasing. We first rewrite $U_{\osc}$ as
\beqq
    U_{\osc}=\frac{\beta_R\left(\frac{\alpha}{\alpha_{\osc}}\right)^{\frac{1}{2}}-\beta_R}{1+\frac{\alpha^{\frac{1}{2}}}{(4 k_D \gamma_{R,\max}V)^{\frac{1}{2}}}-\frac{\gamma_{R,\max}Vc_U}{c_P\gamma_{P,\max}(1-V)+c_R\gamma_{R,\max}V}}.
\eeqq
Using $ -\frac{\alpha^{\frac{1}{2}}}{4( k_D \gamma_{R,\max})^{\frac{1}{2}}}V^{\frac{-3}{2}}-\frac{c_P\gamma_{P,\max}\gamma_{R,\max}c_U}{(c_P\gamma_{P,\max}(1-V)+c_R\gamma_{R,\max}V)^2}<0$,
we obtain that $U_{\osc}$ is an increasing function of $V$.

\noindent
Let's compute the function $\Psi$, using the characteristic polynomial \eqref{eq:charac_poly} and the $\gamma_{i,\max}$ defined in \eqref{eq:gamma_function}. If $\Delta\leq0 $, both roots of the polynomial have the same real value, thus :
\begin{equation}
    \Psi=R^*\frac{c_R\gamma_R'(E^*)V+c_P\gamma_P'(E^*)(1-V)}{2}=R^{*}\frac{c_R \gamma_{R,\max}k_D V + c_P\gamma_{P,\max}k_D(1-V)}{2(k_D+E^*)^2}.
\end{equation}
If $\Delta>0$,
\begin{equation}
    \begin{split}
        \Psi&=\frac{1}{2}\left[R^*(c_R\gamma_R'(E^*)V+c_P\gamma_P'(E^*)(1-V))-\sqrt{\Delta}\right]\\
        &=R^{*}\frac{c_R \gamma_{R,\max}k_D V + c_P\gamma_{P,\max}k_D(1-V)}{2(k_D+E^*)^2}\left(1-\sqrt{1-4\frac{\alpha  \gamma_{R,\max}k_D V(k_D+E^*)^2}{R^{*2}(c_P\gamma_{P,\max}k_D(1-V) +c_R \gamma_{R,\max}k_D V )^2}}\right).
    \end{split}
\end{equation}
We now prove $(i)$ and $(ii)$

\begin{enumerate}[label=(\roman*)]

\item To study $\Psi(V)$, we distinguish the two cases:
\begin{itemize}
    \item If $U\geq U_{\osc}(V)$, then $\Delta\leq0$ so $\Psi=\frac{\alpha f(V)}{(k_D+E^*)^2}$, where $f(V) = \frac{1}{\alpha} R^* (c_P\gamma_{P,\max}(1-V) +c_R \gamma_{R,\max} V)$. Note first, from equation \eqref{eq:E*}, that $(k_D+E^*)^2$ is decreasing with respect to $V$. To study $f(V)$, we first note that if $c_R \gamma_{R,\max}\geq c_P \gamma_{P,\max}$, then $f$ is increasing. Else, we derive
    $$
        \frac{df}{dV}=
         V^2 (c_R \gamma_{R,\max}-c_P \gamma_{P,\max})((c_R \gamma_{R,\max}-c_P \gamma_{P,\max})(\beta_R+U)-c_U U)+c_P^2\frac{\gamma_{P,\max}^2}{\gamma_{R,\max}}(\beta_R+U).
    $$
    In particular, if $c_R \gamma_{R,\max}\geq c_P \gamma_{P,\max}$, then $\frac{df}{dV}>0$, so in both cases,  we find that $f$ is increasing.
    
\item If $U< U_{\osc}$, then 
$$\Psi(V)=R^{*}\frac{c_P\gamma_{P,\max}k_D(1-V) +c_R \gamma_{R,\max}k_D V }{2(k_D+E^*)^2}\left(1-F(V)\right),$$
where
\begin{equation}\label{def:F_psi}
    F(V) = \sqrt{1-4\frac{\alpha  \gamma_{R,\max}k_D V(k_D+E^*)^2}{R^{*2}(c_P\gamma_{P,\max}k_D(1-V) +c_R \gamma_{R,\max}k_D V )^2}}.
\end{equation}
Similarly as in the previous case, we can show that $F(V)$ is increasing. We next introduce 
$$G(V)=\frac{R^{*}(c_P\gamma_{P,\max}(1-V) +c_R \gamma_{R,\max} V )}{\alpha V} = C \frac{1}{\Psi(V)(1+F(V))},$$
where $C$ is a positive constant. Upon differentiating $G$, we obtain

\begin{equation*}
    \frac{dG}{dV}=\frac{c_P\gamma_{P,\max}c_UU}{(c_P\frac{\gamma_{P,\max}}{\gamma_{R,\max}}(\beta_R+U)+(c_R(\beta_r+U)-c_U U-c_P\frac{\gamma_{P,\max}}{\gamma_{R,\max}}(\beta_R+U))V)^2} \geq0,
\end{equation*}
so $\Psi' (1+F) +  F'\Psi \leq 0.$ Using that $F'$, $\Psi$ and $(1+F) \geq 0$, we conclude that $\Psi$ is  decreasing.
\end{itemize}

\item To study $\Psi(U)$, we distinguish the two cases
\begin{itemize}
\item If $U> U_{\osc}$, $\Psi=R^{*}\frac{c_P\gamma_{P,\max}k_D(1-V) +c_R \gamma_{R,\max}k_D V }{2(k_D+E^*)^2} = Cf(U)$, where $C>0$ and
$$f(U) = (\gamma_{R,\max}k_D V)^2\frac{R^{*}}{(k_D+E^*)^2}=\frac{\alpha(\gamma_{R,\max} V-(\beta_R+U))^2}{c_R(\beta_r+U)+c_P\frac{\gamma_{P,\max}(1-V)}{\gamma_{R,\max}V}(\beta_R+U)-c_U U}.$$
In particular, we have 
$$\frac{df}{dU} = \frac{((U+\beta_R)^2- (\gamma_{R,\max} V)^2)(c_R+c_P\frac{\gamma_{P,\max}(1-V)}{\gamma_{R,\max}V}-c_U)+2 (U+\beta_R- \gamma_{R,\max} V)\beta_R c_U}{(c_R(\beta_r+U)+c_P\frac{\gamma_{P,\max}(1-V)}{\gamma_{R,\max}V}(\beta_R+U)-c_U U)^2}.$$
Using \eqref{eq:mass_conservation} and \eqref{eq:constraint2}, $\frac{df}{dU}<0$ so $\Psi$ is a decreasing function of $U$.

\item If $U< U_{\osc}$, $
\Psi(U)=R^{*}\frac{c_P\gamma_{P,\max}k_D(1-V) +c_R \gamma_{R,\max}k_D V }{2(k_D+E^*)^2}\left(1-F(U)\right),$
where $F(U)$ is the same function as in equation \eqref{def:F_psi}. Note that
\beqq
  \frac{k_D+E^*}{R^*}=\frac{k_D V\gamma_{R,\max}}{\alpha(V\gamma_{R,\max}-(\beta_R+U))}\left(\frac{c_P\gamma_{P,\max}(1-V)+c_R\gamma_{R,\max}V}{\gamma_{R,\max}V}(\beta_R+U)-c_U U\right)
\eeqq
is increasing because of the assumption \eqref{eq:mass_conservation}, so $F(U)$ is decreasing. Similarly as in $(i)$, we can show that $(1+F)\Psi' +F'\Psi \geq 0$, and conclude that $\Psi(U)$ is increasing.
\end{itemize}
\end{enumerate}

\end{proof}

\begin{lemma}\label{lemma_lower_bound}
Let $(U,V)\in \left\{U_{\min} , U_{\max} \right\}\times  \left\{ V_{\min} , V_{\max} \right\}$, then
\beqq
\Psi(U,V) \geq \min \left( \Psi(U_{\min},V_{\max}),\Psi(U_{\max},V_{\min})\right).
\eeqq
\end{lemma}

\begin{proof}
If $U\leq U_{\osc}$, and since $U_{\osc}(V)$ is increasing, then proposition \ref{prop:Psi_variation} yields $\Psi(U,V) \geq\Psi(U_{min},V_{max})$. Else, we also conclude from proposition \ref{prop:Psi_variation} that $\Psi(U,V) \geq\Psi(U_{\max},V_{\min})$.
    
\end{proof}

%%%%%%%%%%%%%%%%%%%%%%%%%%%%%%%%%%%%%%%%%%%%%%%%%%%%%%%%%%%%%%%%%%%%%%%%%%
\newpage
%%%%%%%%%%%%%%%%%%%%%%%%%%%%%%%%%%%%%%%%%%%%%%%%%%%%%%%%%%%%%%%%%%%%%%%%%%
\section{Behaviour of the objective function}\label{appendix_3}
%%%%%%%%%%%%%%%%%%%%%%%%%%%%%%%%%%%%%%%%%%%%%%%%%%%%%%%%%%%%%%%%%%%%%%%%%%
In this Appendix section, we prove the following results that are used in section \ref{sec:results_oc} of the main text.

\begin{proposition}\label{prop:rho_star}
Let $\rho^*$ be the protein production rate value at $(R^*,E^*)$, so that
\begin{equation}
\rho^*=\frac{\alpha\gamma_{P,\max}(\beta_R+U)(1-V)}{\gamma_{R,\max}V(c_R(\beta_r+U)+\frac{c_P\gamma_{P,\max}(1-V)(\beta_R+U)}{\gamma_{R,\max}V}-c_U U)}.\label{eq:rho_star}
\end{equation}
Then, $\rho^*$ is an increasing function of $U$, and decreasing function of $V$.
\end{proposition}
\begin{proof}
Differentiating Eq. \eqref{eq:rho_star} with respect to $U$ yields

\begin{equation}
    \frac{d \rho^*}{dU}=\frac{\alpha\gamma_{P,\max}(1-V)}{\gamma_{R,\max}V} \frac{c_U\beta_r }{(c_R(\beta_r+U)+\frac{c_P\gamma_{P,\max}(1-V)(\beta_R+U)}{\gamma_{R,\max}V}-c_U U)^2}>0, 
\end{equation}
so $\rho^*(U)$ is increasing. On the other hand, we can re-write Eq. \eqref{eq:rho_star} as
\begin{equation*}
       \rho^* =\frac{\gamma_{P,\max}\alpha(\beta_R+U)}{\gamma_{R,\max}}\frac{1}{\frac{V}{1-V}(c_R\beta_r+(c_R-c_U)U)+c_P\frac{\gamma_{P,\max}}{\gamma_{R,\max}}(\beta_R+U)}.
\end{equation*}
Since  $x\mapsto \frac{x}{1-x}$ is strictly increasing on $[0,1]$ and $(c_R\beta_r+(c_R-c_U)U)> 0$ (as a result of the mass conservation \eqref{eq:mass_conservation}t), we conclude that $\rho^*$ is a strictly decreasing function of $V$.
\end{proof}

We next prove proposition \ref{prop:cst_ctrl} (stated here again for clarity):

\begin{proposition*}
Let $\rho^*(U,V, \alpha)$ be the value of the production rate function $\rho$ (Eq. \eqref{eq:rho}) at the steady state $(R^*,E^*)$ associated with $(U,V, \alpha)$ (Eq. \eqref{eq:R*}-\eqref{eq:E*}), and assume that the quasi-static approximation holds (Eq. \eqref{eq:Hypothesis_char_time}). Then the optimal control problem Eq. \eqref{eq:optimal_formulation} admits a solution if and only if the equation 

\begin{equation*}
\rho^*(U,V,\alpha_{\max}) = \rho^*(U_{\max}, V_{\min},\alpha_{\min})
\end{equation*}
admits a solution $(U_s,V_s)$. In this case,
\begin{equation*}
    (\hat{U},\hat{V})(t) = \left\{
    \begin{array}{ll}
        (U_{\max},V_{\min})  & \mbox{if  } 0\leq t<\frac{\tau}{2}\\
        (U_s,V_s) & \mbox{if } \frac{\tau}{2} \leq t< \tau
    \end{array}
\right.
\end{equation*}
is optimal for Eq. \eqref{eq:optimal_formulation}, and solutions of Eq. \eqref{eq:usvs} define a 1D-manifold $\Gamma=[U_{\min},U_{\max}]\times[V_{\min},V_{\max}]\cap\{(\Tilde{F}(V),V), V\in[0,1]\}$, such that $\Tilde{F}$ is an increasing function of $V$. 
\end{proposition*}

\begin{proof}
When $\alpha=\alpha_{\min}$, proposition \ref{prop:rho_star} gives us the values for $U$ and $V$ that maximize $\rho^*$, as $U=U_{\max},V=V_{\min}$. The associated production rate is then
\beqq\label{low_prod_rate}
\rho_{0}\coloneqq\frac{\alpha_{\min}\gamma_{P,\max}(\beta_R+U_{\max})(1-V_{\min})}{\gamma_{R,\max}V_{\min}(c_R(\beta_r+U_{\max})+\frac{c_P\gamma_{P,\max}(1-V_{\min})(\beta_R+U_{\max})}{\gamma_{R,\max}V_{\min}}-c_U U_{\max})}.
\eeqq
When $\alpha=\alpha_{\max}$, and for the protein production to stay constant, the controls $U,V$ are then solutions of 
\beqq
\frac{\alpha_{\max}\gamma_{P,\max}(\beta_R+U)(1-V)}{\gamma_{R,\max}V(c_R(\beta_r+U)+\frac{c_P\gamma_{P,\max}(1-V)(\beta_R+U)}{\gamma_{R,\max}V}-c_U U)}=\rho_{0}.
\eeqq
We thus have to solve the following non-linear system
\begin{equation*}
    \begin{split}
    U\left(c_P\gamma_{P,\max}\rho_{0}-\alpha_{\max}\gamma_{P,\max}+V\left[\gamma_{R,\max}c_R\rho_{0}-c_P\gamma_{P,\max}\rho_{0}-\gamma_{R,\max}c_U\rho_{0}+\alpha_{\max}\gamma_{P,\max}\right]\right)\\
    =-V\beta_R\left[\gamma_{R,\max}c_R\rho_{0}-c_P\gamma_{P,\max}\rho_{0}+\alpha_{\max}\gamma_{P,\max}\right] -c_P\gamma_{P,\max}\beta_R\rho_{0}+\alpha_{\max}\gamma_{P,\max}\beta_R
\end{split}
\end{equation*}
Note that the term multiplying $U$ on the left hand side of the equation is not equal to 0. Otherwise, it would lead to have a value of $V$ that is not compatible with $\gamma_{R,\max}c_U\rho_{0}>0$. The following function $\Tilde{F}_{\alpha}(V)$ is hence well defined for $\alpha=\alpha_{\max}$:
\beq\label{eq:tildeF}
\Tilde{F}_{\alpha}(V)=\frac{-V\beta_R\left[\gamma_{R,\max}c_R\rho_{0}-c_P\gamma_{P,\max}\rho_{0}+\alpha\gamma_{P,\max}\right] -c_P\gamma_{P,\max}\beta_R\rho_{0}+\alpha\gamma_{P,\max}\beta_R}{c_P\gamma_{P,\max}\rho_{0}-\alpha\gamma_{P,\max}+V\left[\gamma_{R,\max}c_R\rho_{0}-c_P\gamma_{P,\max}\rho_{0}-\gamma_{R,\max}c_U\rho_{0}+\alpha\gamma_{P,\max}\right]},
\eeq
and the optimal control values  are solutions of $U=\Tilde{F}_{\alpha_{\max}}(V)$. Since $\rho^*$ is an increasing function in $U$, a decreasing function in $V$ and since $\rho^*(\Tilde{F}_{\alpha_{\max}}(V),V,\alpha_{\max})$ is constant for the set of admissible parameters, we deduce that $\Tilde{F}_{\alpha_{\max}}(V)$ is increasing. Note that with the same considerations, we deduce that $\Tilde{F}_{\alpha_{\max}}$ is a decreasing function of $\alpha_{\max}$. Moreover, by construction and because $\alpha_{\max}>\alpha_{\min}$, we know that $\Tilde{F}_{\alpha_{\max}}(V_{\min})<U_{\max}$. In particular, there exists a constant control if and only if $\Tilde{F}_{\alpha_{\max}}(V_{\max})\geq U_{\min}$
\end{proof}

%%%%%%%%%%%%%%%%%%%%%%%%%%%%%%%%%%%%%%%%%%%%%%%%%%%%%%%%%%%%%%%%%%%%%%%%%%
\newpage
%%%%%%%%%%%%%%%%%%%%%%%%%%%%%%%%%%%%%%%%%%%%%%%%%%%%%%%%%%%%%%%%%%%%%%%%%%
\section{Optimal control with one fixed control }\label{appendix_4}
%%%%%%%%%%%%%%%%%%%%%%%%%%%%%%%%%%%%%%%%%%%%%%%%%%%%%%%%%%%%%%%%%%%%%%%%%%
In this Appendix section, we detail the analytical calculations to support the results from section \ref{sec:constraints}. The protein production rate at steady state is 
\begin{equation}
    \rho^*(U,V,\alpha)=\frac{\gamma_{P,\max}}{\gamma_{R,\max}}(U+\beta_R)\frac{1-V}{V}\frac{\alpha}{c_R(\beta_r+U)+\frac{c_P\gamma_{P,\max}(1-V)(\beta_R+U)}{\gamma_{R,\max}V}-c_U U}.
\end{equation}
Assuming $0<V<1$, having a constant protein production over the two periods of the control (Eq. \eqref{eq:constraints_rho}) is equivalent to
\begin{equation}
    \alpha_{\min}=\alpha_{\max}\frac{(1-V_{\max})V_{\min}(c_R+c_P\frac{\gamma_{P,\max}(1-V_{\min})}{\gamma_{R,\max}V_{\min}}-c_U \frac{U}{U+\beta_R})}{(1-V_{\min})V_{\max}(c_R+c_P\frac{\gamma_{P,\max}(1-V_{\max})}{\gamma_{R,\max}V_{\max}}-c_U \frac{U_{\min}}{U_{\min}+\beta_R})}.
\end{equation}
Reordering this equation, we obtain at constant protein production
\begin{equation}\label{eq:u_hat}
    U=\hat{f}_U(U_{\min}), \text{ where $\hat{f}_U$ is defined as } \hat{f}_U(U_{\min})\coloneqq \frac{\beta_R g(U_{\min})}{1-g(U_{\min})},
\end{equation}
where 
\begin{equation}
    \begin{split}
        g(U_{\min})=&\frac{\alpha_{\min}(1-V_{\min})V_{\max}}{\alpha_{\max}(1-V_{\max})V_{\min}}\left(\frac{U_{\min}}{U_{\min}+\beta_R}\right)\\ &+\frac{c_R}{c_U}\left(1-\frac{\alpha_{\min}(1-V_{\min})V_{\max}}{\alpha_{\max}(1-V_{\max})V_{\min}}\right)+\frac{c_P}{c_U}\frac{\gamma_{P,\max}(1-V_{\min})}{\gamma_{R,\max}V_{\min}}\left(1-\frac{\alpha_{\min}}{\alpha_{\max}}\right).
    \end{split}
\end{equation}
Now that we can link $U_{\min}$ and $U$, we  estimate a range of values for $\alpha_{\min}$ that ensure constant protein production. To produce valid values $U$, the following conditions should hold 
\begin{equation}\label{eq:constraint3}
    U_{\min}\leq \hat{f}_U(U_{\min}),\quad \text{ and }\quad \beta_R+\hat{f}_U(U_{\min})\leq\gamma_{R,\max}V_{\min} \;\text{ (see \eqref{eq:constraint2})}.
\end{equation}
Note that $g$ is increasing, so $\hat{f}_U$ is also increasing. We derive now a lower bound for $\alpha_{\min}$: for a solution to Eq.\eqref{eq:constraint3} to exist, $\hat{f}_U$ has to be less than $\gamma_{R,\max}V_{\min}-\beta_R$ for some value. Thus from
\begin{equation}\label{eq:bound_alpha1}
    \hat{f}_U(0)\leq \gamma_{R,\max}V_{\min}-\beta_R,
\end{equation}
we deduce the lower bound :
\begin{equation}\label{eq:lower_bound_alpha}
    \frac{\alpha_{\min}}{\alpha_{\max}}\geq q_1\coloneqq\frac{(c_R-c_U)\gamma_{R,\max}V_{\min}+c_P\gamma_{P,\max}(1-V_{\min})+c_U\beta_R}{(1-V_{\min})\left(c_R\gamma_{R,\max}\frac{V_{\max}}{1-V_{\max}}+c_P\gamma_{P,\max}\right)}>0.
\end{equation}
Since $\hat{f}_U$ has to be positive to verify Eq.\eqref{eq:constraint3}, such that at least
\begin{equation}\label{eq:bound_alpha2_app}
    \hat{f}_U(\gamma_{R,\max}V_{\min}-\beta_R)\geq 0
\end{equation}
we derive the upper bound

\begin{equation}\label{eq:upper_bound_alpha_app}
    \frac{\alpha_{\min}}{\alpha_{\max}}\leq q_2\coloneqq\frac{c_R\gamma_{R,\max}V_{\min}+c_P\gamma_{P,\max}(1-V_{\min})+c_U\frac{\alpha(1-V_{\min})V_{\max}}{\alpha_{\max}(1-V_{\max})V_{\min}}(\gamma_{R,\max}V_{\min}-\beta_R)}{(1-V_{\min})\left(c_R\gamma_{R,\max}\frac{V_{\max}}{1-V_{\max}}+c_P\gamma_{P,\max}\right)}>0.
\end{equation}

We now study the value of control bounds $U_{\max},V_{\min}$ enabling the existence of an optimal solution of problem \eqref{eq:optimal_formulation}, for a given protein production  $\rho_0$ at high resource intake $\alpha=\alpha_{\max}$. We use the properties of function $\Tilde{F}_{\alpha}$ detailed in the proof of Proposition \ref{prop:cst_ctrl} (Appendix \ref{appendix_3}). In this case, the domain of $(U_{\max},V_{\min})$ that allows to find optimal controls is more simple to obtain, as it is given by
\begin{equation}\label{eq:region_exist_ctrl}
    \{(U_{\max},V_{\min})\quad|\quad \rho^*(U_{\max},V_{\min},\alpha_{\min})\geq \rho_0,\; U_{\max}\geq U_{\min},\; V_{\min}\leq V_{\max} \}.
\end{equation}
In particular, solving $\rho^*(U,V,\alpha_{\min})= \rho_0$ yields an explicit expression for $U= \Tilde{F}_{\alpha_{\min}}(V)$, where  $\Tilde{F}_{\alpha}$ is defined in Appendix \ref{appendix_3}, equation \eqref{eq:tildeF}. Note that while the corresponding limit value for $\alpha_{\min}$ cannot be explicitly derived in general (non trivial optimization), we obtain a lower bound by solving
 \begin{equation}
     \Tilde{F}_{\alpha_{\min}}\left(\frac{U_{\min}+\beta_R}{\gamma_{R,\max}}\right)\leq \gamma_{R,\max}V_{\max}-\beta_R,
 \end{equation}
 which yields
 \begin{equation}\label{eq:domain_UV_app}
    \alpha_{\min}\geq\rho_{0}q_3 \coloneqq\left(\frac{U_{\min}+\beta_R}{\gamma_{R,\max}-(U_{\min}+\beta_R)}\left(c_R\frac{\gamma_{R,\max}}{\gamma_{P,\max}}-c_U\frac{\gamma_{R,\max}V_{\max}-\beta_R}{\gamma_{P,\max}V_{\max}}\right)+c_P\right)>0.
\end{equation}

%%%%%%%%%%%%%%%%%%%%%%%%%%%%%%%%%%%%%%%%%%%%%%%%%%%%%%%%%%%%%%%%%%%%%%%%%%
\newpage
%%%%%%%%%%%%%%%%%%%%%%%%%%%%%%%%%%%%%%%%%%%%%%%%%%%%%%%%%%%%%%%%%%%%%%%%%%
\section{Biophysical parameters}\label{appendix_5}
%%%%%%%%%%%%%%%%%%%%%%%%%%%%%%%%%%%%%%%%%%%%%%%%%%%%%%%%%%%%%%%%%%%%%%%%%%

The parameters used in our numerical application are given in Table \ref{table:parameters}, with values matching with experimental data collected for eukaryotic cells when available, and {\it E. coli} if not. The ribosomal loss rate $\beta_R$ was directly derived from the half life of its main components, ribosomal RNA and proteins (5 days \cite{hirsch_turnover_1966}). The conversion factors $c_R,c_P$ are the average number of amino-acids in a ribosome and in a typical protein in {\it E. coli}, respectively  \cite{scott_emergence_2014}, while $c_U$ represents the number of amino-acids that can be recycled upon targeted ribosome degradation, times an efficiency factor that we set at 0.9. The maximal and minimal fractions of actively translating ribosomes producing ribosomal or generic proteins $V_{\max},V_{\min}$
are taken from \cite{scott_emergence_2014}. The parameters of the function $\gamma_i$  are taken from Klump et {\it al.}~\cite{klumpp_molecular_2013}, where the elongation speed is $$ \gamma=\gamma_{max}\frac{c_T}{c_T+K_M},$$
with $c_T$ the concentration of translation-affiliated proteins (EF-Tu, EF-G, EF-Ts, tRNA synthetases, etc), and $K_M$ the Michaelis constant, $K_M\approx 3\mu M$. We suppose that the concentration of those proteins are proportional to the amino-acid concentration: $c_T= n_0\varphi_0 E$, with $n_0$ a caracteristic number of amino acid in those proteins, and $\varphi_0$ the fraction of amino acids in tRNA complexes, where we set $n_0=400$ \cite{encyclopedia-2004} and $\varphi_0=0.8$ \cite{zhang_generic_2009}. Finally, we obtain $k_D=K_M n_0=7.2\cdot10^5\ /\mu m^3$. The maximum elongation rate $\gamma_{i,\max}$ is computed with the number of amino-acid of the protein and the maximum speed of translation of an amino-acid $ k_{elong}= 30AA/s$ \cite{klumpp_molecular_2013}.

\begin{table}[ht!]
 \caption{\textbf{Table of parameters}}
 \vspace{1em}
 \label{table:parameters}
 \centering
\begin{tabular}{ |c c c |}
 \hline
Parameter & Description & Value\\
  \hline
 $\beta_R$ & Ribosomes degradation rate &  $6\cdot 10^{-3}\;h^{-1}$ \;\cite{hirsch_turnover_1966} \\ &&\\
 $V_{\min}\; - \;V_{\max}$ &Extremal values of $V$ & $0.6\; - \;0.791\in [0.45,0.95]$ \;\cite{scott_emergence_2014} \\ &&\\
 $U_{\min}\; - \;U_{\max}$ & Extremal values of $U$ & $4\; - \;8\;h^{-1}$  \\ &&\\
 $c_P$ & Conversion factor from $E$ to $P$ & $300\in [50,2000]$  \;\cite{Albe_2002_book} \\ &&\\
 $c_R$ & Conversion factor from $E$ to $R$ & $7736$\;\cite{dai_reduction_2017} \\ &&\\
 $c_U$ & Ribophagy conversion factor & $7736*0.9$ \;\cite{dai_reduction_2017} \\ &&\\
 $\gamma_{P,\max}$ & Maximal protein production rate & $300\in$[54,2160] $h^{-1}$ \;\cite{klumpp_molecular_2013}\\ &&\\
 $ \gamma_{R,\max}$ & Maximal ribosome production rate & 15 $h^{-1}$ \;\cite{klumpp_molecular_2013} \\ &&\\
 $\tau$ & Period of the $\alpha$ function & 24 $h$\\ &&\\
 $k_D$ & \makecell[t]{Michaelis-Menten constant\\ for the elongation rate }& $7.2\cdot10^5\ \mu m^{-3}$ \;\cite{dai_reduction_2017,klumpp_molecular_2013}\\ &&\\
 $\alpha_{\min}\; - \;\alpha_{\max}$ & Extremal values of $\alpha$ & \makecell[t]{$2\cdot 10^{8} \;- \;2.5\cdot 10^{8}$ \\ $\in[4\cdot10^7,4\cdot10^{11}]\; h^{-1}\mu m^{-3}$ \;\cite{belliveau_fundamental_2020} } \\

 \hline
 
\end{tabular}
\end{table}

\end{appendices}
\end{document}